%% file: revised-gleich-kloster.tex
\documentclass[10pt]{article}

\input{docsetup}
\input{preamble}
\usepackage{multicol}

\usepackage{enumitem}

\newcommand{\floor}[1]{\lfloor #1 \rfloor }
\newcommand{\lp}{\left( }
\newcommand{\rp}{\right) }
\newcommand{\lrp}[1]{\left( #1 \right)}

\newcommand{\eoe}[2]{\ve_{#1}\kron \ve_{#2}}
\newcommand{\inv}{^{-1}}

\newcommand{\hvv}{\hat{\vv}}

\newcommand{\hvxj}[2]{\ensuremath{\hat{\boldsymbol{#1}}^{(#2)}}}
\newcommand{\vxj}[2]{\ensuremath{\boldsymbol{#1}^{(#2)}}}
\newcommand{\nnz}[1]{\mbox{nnz}(#1)}

\newcommand{\gexpm}{\texttt{gexpm}\xspace}
\newcommand{\gexpmq}{\texttt{gexpmq}\xspace}
\newcommand{\expmimv}{\texttt{expmimv}\xspace}
\newcommand{\expmv}{\texttt{expmv}\xspace}

\newcommand{\expof}[1]{\exp\left\{#1\right\}}
\newcommand{\epec}{\expof{\mP}\ve_c}

\newcommand{\blvec}[1]{\bmat{ #1_0 \\ #1_1 \\ \vdots \\ \vdots \\ #1_N }}

\renewcommand{\abstract}{}

\title{Sublinear Column-wise Actions of the\\
Matrix Exponential on Social Networks
}
\author{David F.~Gleich \\ Computer Science Department \and Kyle Kloster \\ Mathematics Department \and
Purdue University\\
\{dgleich, kkloste\}@purdue.edu
}

\begin{document}

\maketitle

\begin{abstract}
We consider stochastic transition matrices from large social and information networks. For these matrices, we describe and evaluate three fast methods to estimate one column of the matrix exponential. The methods are designed to exploit the properties inherent in social networks, such as a power-law degree distribution. Using only this property, we prove that one of our algorithms has a sublinear runtime. We present further experimental evidence showing that all of them run quickly on social networks with billions of edges and accurately identify the largest elements of the column.
\end{abstract}

\input{sec-intro} 

\input{sec-background} 

\input{sec-algorithms} 

\input{sec-analysis-dgleich} 

\input{sec-powerlaw} 

\input{sec-experiments} 

\input{sec-conclusions} 

\section*{Acknowledgments}
This research was supported by NSF CAREER award 1149756-CCF.

\bibliographystyle{plainnat}
\bibliography{all-bibliography,newbib}
\normalsize
\appendix
\input{sec-appendix} 

\end{document}

%% file: docsetup.tex
\usepackage{ragged2e}
\usepackage{geometry}
\makeatletter
\renewcommand\normalsize{%
   \@setfontsize\normalsize\@xpt{14}%
   \abovedisplayskip 10\p@ \@plus2\p@ \@minus5\p@
   \abovedisplayshortskip \z@ \@plus3\p@
   \belowdisplayshortskip 6\p@ \@plus3\p@ \@minus3\p@
   \belowdisplayskip \abovedisplayskip
   \let\@listi\@listI}
\normalsize
\renewcommand\small{%
   \@setfontsize\small\@ixpt{12}%
   \abovedisplayskip 8.5\p@ \@plus3\p@ \@minus4\p@
   \abovedisplayshortskip \z@ \@plus2\p@
   \belowdisplayshortskip 4\p@ \@plus2\p@ \@minus2\p@
   \def\@listi{\leftmargin\leftmargini
               \topsep 4\p@ \@plus2\p@ \@minus2\p@
               \parsep 2\p@ \@plus\p@ \@minus\p@
               \itemsep \parsep}%
   \belowdisplayskip \abovedisplayskip
}
\renewcommand\footnotesize{%
   \@setfontsize\footnotesize\@viiipt{10}%
   \abovedisplayskip 6\p@ \@plus2\p@ \@minus4\p@
   \abovedisplayshortskip \z@ \@plus\p@
   \belowdisplayshortskip 3\p@ \@plus\p@ \@minus2\p@
   \def\@listi{\leftmargin\leftmargini
               \topsep 3\p@ \@plus\p@ \@minus\p@
               \parsep 2\p@ \@plus\p@ \@minus\p@
               \itemsep \parsep}%
   \belowdisplayskip \abovedisplayskip
}
\renewcommand\scriptsize{\@setfontsize\scriptsize\@viipt\@viiipt}
\renewcommand\tiny{\@setfontsize\tiny\@vpt\@vipt}
\renewcommand\large{\@setfontsize\large\@xipt{15}}
\renewcommand\Large{\@setfontsize\Large\@xiipt{16}}
\renewcommand\LARGE{\@setfontsize\LARGE\@xivpt{18}}
\renewcommand\huge{\@setfontsize\huge\@xxpt{30}}
\renewcommand\Huge{\@setfontsize\Huge{24}{36}}
\makeatother

\usepackage[tableposition=top]{caption}

%% file: preamble.tex
\usepackage{ifpdf}
\usepackage{xcolor}
\usepackage{paralist}
\usepackage{tabularx}
\usepackage{booktabs}
\usepackage{multirow}
\usepackage{graphicx}
\usepackage{sidecap}
\usepackage{natbib}
\renewcommand{\cite}{\citep}
\usepackage{amssymb}
\usepackage{amsmath}
\usepackage{lscape}
\makeatletter
\def\vcdots{\vbox{\baselineskip4\p@ \lineskiplimit\z@
    \kern3\p@\hbox{.}\hbox{.}\hbox{.}\kern3\p@}}
\makeatother
\usepackage{caption}
\usepackage{subfigure}
\usepackage{array}
\usepackage{url}
\usepackage{float}

\ifpdf
  \usepackage{epstopdf}
	\usepackage{microtype}
\fi
\graphicspath{{./}{./figures/}}

\usepackage{algorithm}
\usepackage[noend]{algorithmic}

\colorlet{TufteRed}{red!80!black}

\definecolor{theblue}{RGB}{0,0,180}
\colorlet{thered}{TufteRed}

\usepackage{dgleich-setup}
\usepackage[notheorems]{dgleich-math}

\renewcommand{\mat}{\boldsymbol}

\usepackage{mathtools}
\usepackage[amsmath,thmmarks]{ntheorem}

\usepackage{etoolbox}
\apptocmd{\sloppy}{\hbadness 10000\relax}{}{}
\apptocmd{\thebibliography}{\raggedright}{}{}

\input{newtheorems.tex}

\ifpdf%
  \RequirePackage[colorlinks,pdfdisplaydoctitle
      ,citecolor=theblue
      ,linkcolor=theblue
      ,urlcolor=theblue
      ,hyperfootnotes=false,pdftex,pagebackref]{hyperref}%
\else
  \RequirePackage[colorlinks,pdfdisplaydoctitle
      ,citecolor=theblue
      ,linkcolor=theblue
      ,urlcolor=theblue
      ,hyperfootnotes=false,dvipdfm,pagebackref]{hyperref}
\fi
\renewcommand*{\backref}[1]{}
\renewcommand*{\backrefalt}[4]{%
  \ifcase #1 %
    No citations.%
  \or
    Cited on page #2.%
  \else
    Cited on pages #2.%
  \fi
}

\newcommand{\secref}[1]{\hyperref[sec:#1]{section~\ref*{sec:#1}}}
\newcommand{\Secref}[1]{\hyperref[sec:#1]{Section~\ref*{sec:#1}}}
\newcommand{\secrefp}[1]{\hyperref[sec:#1]{(section~\ref*{sec:#1})}}

\newcommand{\thmref}[1]{\hyperref[thm:#1]{theorem~\ref*{thm:#1}}}
\newcommand{\Thmref}[1]{\hyperref[thm:#1]{Theorem~\ref*{thm:#1}}}
\newcommand{\thmrefp}[1]{\hyperref[thm:#1]{(theorem~\ref*{thm:#1})}}

\newcommand{\figrefp}[1]{\hyperref[fig:#1]{(figure~\ref*{fig:#1})}}
\newcommand{\figref}[1]{\hyperref[fig:#1]{figure~\ref*{fig:#1}}}
\newcommand{\Figref}[1]{\hyperref[fig:#1]{Figure~\ref*{fig:#1}}}

%% file: newtheorems.tex
\theoremstyle{break}
\theoremheaderfont{\normalfont\bfseries\itshape}
\theoremprework{\vspace{7pt}\list{}{\rightmargin\leftmargin}%
   \item\relax}
\theorempostwork{\endlist\vspace{7pt}}
\theorembodyfont{\normalfont}
\theoremseparator{}

\newtheorem{theorem}{Theorem}

\theoremclass{theorem}
\newtheorem{lemma}[theorem]{Lemma}
\theoremclass{theorem}

\theoremclass{theorem}

\theoremclass{theorem}
\newtheorem{corollary}[theorem]{Corollary}
\theoremclass{theorem}

\theoremclass{theorem}

\theoremstyle{nonumberplain}
\theoremheaderfont{\normalfont\itshape}
\theorembodyfont{\upshape}
\theoremseparator{\hspace{1ex}}
\theoremsymbol{\rule{1ex}{1ex}}
\theoremprework{\vspace{-7pt}}
\theoremclass{proof}

\newtheorem{proof}{Proof}

\theoremclass{proof}
\theoremstyle{plain}
\theoremheaderfont{\normalfont\bfseries\sscshape}
\theorembodyfont{\upshape}
\theoremseparator{\hspace{1ex}}
\theoremsymbol{\ensuremath{\filleddiamond}}

%% file: sec-intro.tex
\section{Introduction}\label{sec:int}

Matrix exponentials are used for node centrality~\cite{Estrada-2000-index,farahat2006-exphits,Estrada-2010-matrix-functions}, link prediction~\cite{Kunegis-2009-learning-spectral}, graph kernels~\cite{Kondor-2002-diffusion}, and clustering~\cite{chung2007-pagerank-heat}. In the majority of these problems, only a rough approximation of a column of the matrix exponential is needed. Here we present methods for fast approximations of 
\[ \expof{\mP}\ve_c, \]
where $\mP$ is a column-stochastic matrix and $\ve_c$ is the $c$th column of the identity matrix. This suffices for many applications and also allows us to compute $\expof{-\hat{\mL}} \ve_c$ where $\hat{\mL}$ is the normalized Laplacian.

To state the problem precisely and fix notation, let $\mG$ be a graph adjacency matrix of a directed graph and let $\mD$ be the diagonal matrix of out-degrees, where $\mD_{ii} = d_i$, the degree of node $i$. For simplicity, we assume that all nodes have positive out-degrees, thus, $\mD$ is invertible. The methods we present are designed to work for $\mP = \mG\mD\inv$ and, by extension, the negative normalized Laplacian $-\hat{\mL} = \mD^{-1/2}\mG\mD^{-1/2} - \mI$. This is because the relationship
\[
\expof{\mD^{-1/2}\mG\mD^{-1/2} - \mI} = e\inv \mD^{-1/2} \expof{\mG\mD\inv} \mD^{1/2}
\]
implies $\expof{-\hat{\mL}}\ve_c = \sqrt{d_c}e\inv \mD^{-1/2}\expof{\mP}\ve_c$, which allows computation of either column given the other, at the cost of scaling the vector.

\subsection{Previous work}

Computing the matrix exponential for a general matrix $\mA$ has a rich
and ``dubious" history~\cite{Moler-2003-19}.  For any matrix $\mA \in
\mathbb{R}^{n \times n}$ and vector $\vb \in \mathbb{R}^{n}$, one
approach is to use a Taylor polynomial approximation
\[
\expof{\mA}\vb \approx \sum_{j=0}^N \tfrac{1}{j!} \mA^j \vb.
\]
This sequence converges to the correct vector as $N \to \infty$ for
any square matrix, however, it can be problematic numerically.   A
second approach is to first compute an $m \times m$ upper-Hessenberg
form of $\mA$, $\mH_m$, via an $m$-step Krylov method, $\mA \approx
\mV_m \mH_m \mV_m^T$. Using this form, we can approximate
$\expof{\mA}\vb \approx \mV_m \expof{\mH_m}\ve_1$ by performing
$\expof{\mH_m}\ve_1$ on the much smaller, and better controlled,
upper-Hessenberg matrix $\mH_m$. These concepts underlie many standard
methods for obtaining $\expof{\mA}\vb$.

Although the Taylor and Krylov approaches are fast and accurate -- see
references \cite{hoch-lubi-97}, \cite{gall-saad-92}, and
\cite{Al-Mohy-2011-exponential} for the numerical analysis -- existing
implementations depend on repeated matrix-vector products with the
matrix $\mA$. The Krylov-based algorithms also require
orthogonalization steps between successive vectors.
When these algorithms are used to compute exponentials of graphs with
small diameter, like the social networks we consider here, the
repeated matrix-vector products cause the vectors involved to become
dense after only a few steps. The subsequent matrix-vector products
between the sparse matrix and dense vector require $O(|E|)$ work,
where $|E|$ is the number of edges in the graph (and there are
$O(|E|)$ non-zeros in the sparse matrix).
This leads to a runtime bound of $O(T|E|)$ if there are $T$ matrix
vector products after the vectors become dense.

There are a few recent improvements to the Krylov methods that reduce
the number of terms $T$ that must be
used~\cite{Sidje-1998-ExpoKit,Orecchia-2012-exponential,Afanasjew-2008-restarted,Al-Mohy-2011-exponential}
or present additional special cases~\cite{Benzi-2010-quadrature}. Both
\citet{Orecchia-2012-exponential} and \citet{Al-Mohy-2011-exponential}
present a careful bound on the maximum number of terms $T$.  
\Citet{Orecchia-2012-exponential} presents a new
polynomial approximation for $\exp x$ that improves on the Taylor
polynomial approach and uses this to give a tight bound on the
necessary number of matrix-vector products in the case of a general
symmetric positive semidefinite matrix $\mA$. \Citet{Al-Mohy-2011-exponential} presents
a bound on the number of Taylor terms for a matrix with bounded
norm.

Thus, the best runtimes provided by existing methods are $O(|E|)$ for
the stochastic matrix of a graph. It should be noted, however, that
the algorithms we present in this paper operate in the specific
context of matrices with 1-norm bounded by 1, and where the vector
$\vb$ is sparse with only 1 non-zero. In contrast, the existing
methods we mention here apply more broadly.

In the case of exponentials of sparse graphs, Chung and Simpson
developed a Monte Carlo procedure to estimate
columns~\cite{Chung-2013-solving}, like our method. They show that
only a small number of random walks are needed to compute reasonably
accurate solutions; however, the number of walks grows quickly with
the desired accuracy. They also prove their algorithm runs in time
polylogarithmic in $1/\eps$, although the $\eps$ accuracy is achieved
in a degree-weighted infinity norm, making the
computational goal distinct from our own. Our accuracy result is in the $1$-norm,
which provides uniform control over the error.

We note that for a general sparse graph it is impossible to get a work
bound that is better than $O(n)$ for computing $\expof{\mA}\ve_c$ with accuracy $\eps$ in the
1-norm, even if $\mA$ has only $O(n)$
nonzeros. For example, the star graph on $n$ nodes requires
$\Omega(n)$ work to compute certain columns of its exponential, as
they have $O(n)$ nonzero entries of equal magnitude and hence cannot
be approximated with less than $O(n)$ work. This shows there cannot be
a sublinear upperbound on work for accurately approximating general
columns of the exponential of arbitrary sparse graphs.

We are able to obtain our sublinear work bound by assuming structure
in the degree distribution of the underlying graph. Another case where
it is possible to show that sublinear algorithms are possible is when
the matrices are banded, as considered by \citet{Benzi-2007-decay}.
Banded matrices correspond to graphs that look like the line-graph
with up to $d$ connections among neighbors. If $d$ is sufficiently
small, or constant, then the exponential localizes and sublinear
algorithms are possible. However, this case is unrealistic for social
networks with highly skewed degree distributions.

\subsection{Our contributions}
For networks with billions of edges, we want a procedure that avoids the dense vector operations involved in Taylor- and Krylov-based methods. In particular, we would like an algorithm to estimate a column of the matrix exponential that runs in time proportional to the number of \emph{large} entries. Put another way, we want an algorithm that is \emph{local} in the graph and produces a \emph{local solution}.

Roughly speaking, a \textit{local solution} is a vector that both accurately approximates the true solution, which can be dense, \emph{and} has only a few non-zeros. A local algorithm is then a method for computing such a solution that requires work proportional to the size of the solution, rather than the size of the input. In the case of computing a column of the matrix exponential for a network with $|E|$ edges, the input size is $|E|$, but the desired solution $\epec$ has only a few significant entries. An illustration of this is given in Figure~\ref{fig:example}. From that figure, we see that the column of the matrix exponential has about 5 million non-zero entries. However, if we look at the approximation formed by the largest $3,000$ entries, it has a 1-norm error of roughly $10^{-4}$. A local algorithm should be able to find these $3,000$ non-zeros without doing work proportional to $|E|$. For this reason, local methods are a recognized and practical alternative to Krylov methods for solving massive linear systems from network problems; see, for instance, references~\cite{andersen2006-local,Bonchi-2012-fast-katz}. The essence of these methods is that they replace whole-graph matrix-vector products with targeted column-accesses; these correspond to accessing the out-links from a vertex in a graph structure. 

\begin{figure}
\centering
\includegraphics[width=0.5\linewidth]{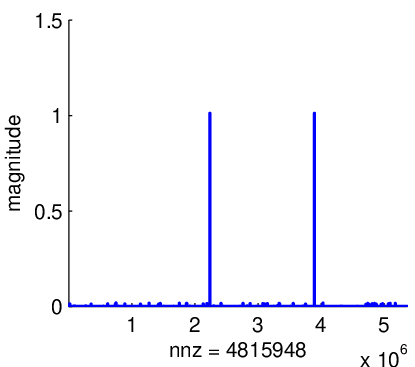}%
\includegraphics[width=0.5\linewidth]{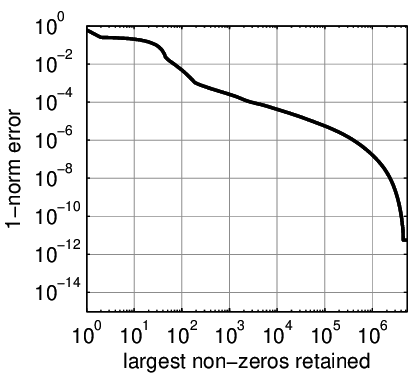}
\caption{(At left) A column of the matrix exponential from the livejournal graph with 5M vertices and 78M directed edges shows only two large entries and a total of 4.8M numerically non-zero entries.  (At right) The second figure shows the solution error in the 1-norm as only the largest entries are retained. This shows that there is a solution with 1-norm error of $10^{-4}$ with around 3,000 non-zeros entries. These plots illustrate that the matrix exponential can be localized in a large network and we seek local algorithms that will find these solutions without exploring the entire graph.}
\label{fig:example}

\includegraphics{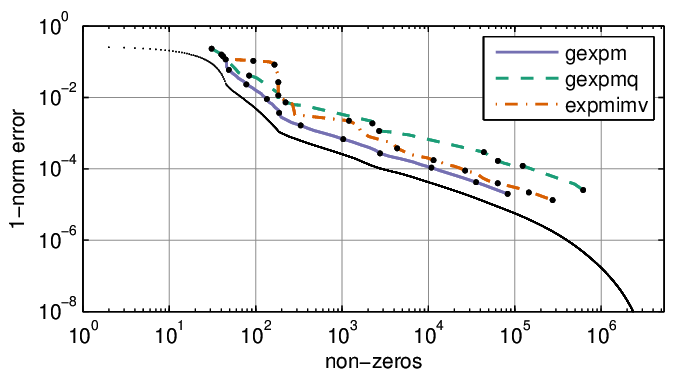}
\caption{The result of running our three algorithms to approximate the vector studied in Figure~\ref{fig:example}. The small black dots show the optimal set of non-zeros chosen by \emph{sorting} the true vector.  The three curves show the results of running our algorithms as we vary the desired solution tolerance. Ideally, they would follow the tiny black dots exactly. Instead, they closely approximate this optimal curve with \gexpm having the best performance. }
\label{fig:localized-vectors}
\end{figure}

In this paper, we present three algorithms that approximate a specified column of $\expof{\mP}$ where $\mP$ is a sparse matrix satisfying $\|\mP\|_1 \leq 1$ (Section~\ref{sec:alg}). The main algorithm we discuss and analyze uses coordinate relaxation (Section~\ref{sec:gs}) on a linear system to approximate a degree $N$ Taylor polynomial (Section~\ref{sec:tay}). This coordinate relaxation method yields approximations guaranteed to satisfy a prescribed error $\eps$. For arbitrary graphs with maximum degree $d$, the error after $l$ iterations of the algorithm we call \gexpm is bounded by $O(l^{-1/(2d)})$ as shown in Theorem~\ref{thm:conv:gexpm}. Given an input error $\eps$, the runtime to produce a solution vector with 1-norm error less than $\eps$ is thus sublinear in $n$ for graphs with $d \leq O(\log \log n)$ as shown in our prior work~\cite{Kloster-2013-nexpokit}.

This doubly logarithmic scaling of the maximum degree is unrealistic for social and information networks, where highly skewed degree distributions are typical. Therefore, in Section~\ref{sec:anl:pow} we consider graphs with a power-law degree distribution, a property ubiquitous in social networks \cite{Faloutsos-1999-power-law,barabasi1999-scaling}. By using this added assumption, we can show that for a graph with a particular power-law distribution, maximum degree $d$, and minimum degree $\delta$, the \gexpm algorithm produces a 1-norm error of $\eps$ in work that scales roughly as $d^2\log(d)^2$, and with total work bounded by $O(\log(1/\eps) \left(1/\eps\right)^{3\delta/2} d^2\log(d) \max\{ \log(d), \log(1/\eps) \} )$ (Theorem~\ref{thm:work}). As a corollary, this theorem \emph{proves} that columns of $\epec$ are localized. 

Our second algorithm, \gexpmq, is a faster heuristic approximation of this algorithm. It retains the use of coordinate descent, but changes the choice of coordinate to relax to something that is less expensive to compute. It retains the rigorous convergence guarantee but loses the runtime guarantee. 

The final method, \expmimv, differs from the first two and does not use coordinate relaxation. Instead, it uses \emph{sparse matrix-vector products} with only the $z$ largest entries of the previous vector to avoid fill-in. This leads to a guaranteed runtime bound of $O(d z \log z)$, discussed in Theorem~\ref{thm:expmimv}, but with no accuracy guarantee. Our experiments in Section~\ref{sec:exp:run} show that this method is orders of magnitude faster than the others. 

Figure~\ref{fig:localized-vectors} compares the results of these algorithms on the same graph and vector from Figure~\ref{fig:example} as we vary the desired solution tolerance $\eps$ for each algorithm. These results show that the algorithms all track the optimal curve, and sometimes closely! In the best case, they compute solutions with roughly three times the number of non-zeros as in the optimal solution; in the worst case, they need about 50 times the number of non-zeros. In the interest of full disclosure, we note that we altered the algorithms slightly for this figure. Namely, we removed a final step that significantly increases the number of non-zeros by making many tiny updates to the solution vector; these updates are so small that they do not alter the accuracy by more than a factor of 2. We also fixed an approximation parameter based on the Taylor degree to aid comparisons as we varied $\eps$. 

As we finish our introduction, let us note that the source code for all of our experiments and methods is available online.\footnote{\url{https://www.cs.purdue.edu/homes/dgleich/codes/nexpokit/}} The remainder of the paper proceeds in a standard fashion by establishing the formal setting (Section~\ref{sec:bkg}), then introducing our algorithms (Section~\ref{sec:alg}), analyzing them (Section~\ref{sec:anl}, Section~\ref{sec:anl:pow}), and then showing our experimental evaluation (Section~\ref{sec:exp}). This paper extends our conference version~\cite{Kloster-2013-nexpokit} by adding the theoretical analysis with the power-law, presenting the \expmimv method, and tightening the convergence criteria for \gexpmq. Furthermore, we conduct an entirely new set of experiments on graphs with billions of edges.

%% file: sec-background.tex
\section{Background}\label{sec:bkg} 
The algorithm that we employ utilizes a Taylor polynomial approximation of $\epec$. 
Here we provide the details of the Taylor approximation for the exponential of a general matrix. We also review the coordinate relaxation method we use in two of our algorithms.

Although the algorithms presented in subsequent sections are designed to work for $\mP$, much of the theory in this section applies to any matrix $\mA \in \mathbb{R}^{n \times n}$. Thus, we present it in its full generality. For sections in which the theory is restricted to $\mP$, we explicitly state so. Our rule of thumb is that we will use $\mA$ as the matrix when the result is general and $\mP$ when the result requires properties specific to our setting.

\subsection{Approximating with Taylor Polynomials}\label{sec:tay}
The Taylor series for the exponential of a matrix $\mA \in \mathbb{R}^{n\times n}$ is given by
\[
\expof{\mA} = \mI + \tfrac{1}{1!}\mA^1 + \tfrac{1}{2!}\mA^2 + \cdots + \tfrac{1}{k!}\mA^k + \cdots
\]
and it converges for any square matrix $\mA$. By truncating this infinite series to $N$ terms, we may define 
\[ T_N(\mA) := \sum_{j=0}^N \tfrac{1}{j!}\mA^j, \] 
and then approximate $\expof{\mA} \vb \approx T_N(\mA) \vb$. For general $\mA$ this polynomial approximation can lead to inaccurate computations if $\|\mA\|$ is large and $\mA$ has oppositely signed entries, as the terms $\mA^j$ can then contain large, oppositely-signed entries that cancel only in exact arithmetic. However, our aim is to compute $\epec$ specifically for a matrix of bounded norm, that is, $\|\mP\|_1 \leq 1$. In this setting, the Taylor polynomial approximation is a reliable and accurate tool. What remains is to choose the degree $N$ to ensure the accuracy of the Taylor approximation makes $ \| \epec - T_N(\mP) \ve_c \| $ as small as desired.

\paragraph{Choosing the Taylor polynomial degree}
Accuracy of the Taylor polynomial approximation requires a sufficiently large Taylor degree, $N$. On the other hand, using a large $N$ requires the algorithms to perform more work. A sufficient value of $N$ can be obtained algorithmically by exactly computing the number of terms of the Taylor polynomial required to compute $\exp(1)$ with accuracy $\eps$. Formally: 
\[ N = \arg \min_k \left\{ k \text{ where } \left( e - \sum_{\ell = 0}^k \frac{1}{\ell !} \right) \le \eps \right\}. \] 
We provide the following simple upper bound on $N$: 
\begin{lemma}\label{lem:Nbound}
Let $\mP$ and $\vb$ satisfy $\|\mP \|_1, \|\vb\|_1 \leq 1$. Then choosing the degree, $N$, of the Taylor approximation, $T_N(\mP)$, such that $N \geq 2\log(1/\eps)$ and $N \geq 3$ will guarantee
\[
\|\expof{\mP} \vb - T_N(\mP)\vb \|_1 \leq \eps
\]
\end{lemma}
We present the proof, which does not inform our current exposition, in Appendix~\ref{sec:app}. Because Lemma~\ref{lem:Nbound} provides only a loose bound, we display in Table \ref{tab:taylordeg} values of $N$ determined via explicit computation of $\exp(1)$, which are tight in the case that $\|\mP\|_1 = 1$.

\begin{table}
\centering
\caption{Choosing the degree, $N$, of a Taylor polynomial to ensure an accuracy of $\eps$ shows that the bound from Lemma~\ref{lem:Nbound} is not tight, and that both methods are slowly growing.} \label{tab:taylordeg}
\begin{tabularx}{0.5\linewidth}{XXX} \toprule
$\eps$ desired & $N$  predicted by Lemma~\ref{lem:Nbound} & $N$ required \\ \midrule
$10^{-5}$ & 24 & 8 \\
$10^{-10}$ & 46 & 13 \\
$10^{-15}$ & 70 & 17 \\
 \bottomrule
\end{tabularx}
\end{table}

\subsection{Error from Approximating the Taylor Approximation}
The methods we present in Section~\ref{sec:alg} produce an approximation of the Taylor polynomial expression $T_N(\mP)\ve_c$, which itself approximates $\epec$. Thus, a secondary error is introduced. Let $\vx$ be our approximation of $T_N(\mP)\ve_c$. We find
\[
\| \epec - \vx \| \leq \| \epec - T_N(\mP)\ve_c \| + \| T_N(\mP)\ve_c - \vx \|,
\]
by the triangle inequality. Lemma~\ref{lem:Nbound} guarantees the accuracy of only the first term; so if the total error of our final approximation $\vx$ is to satisfy $\| \epec - \vx \|_1 \leq \eps $, then we must guarantee that the right-hand summand is less than $\eps$. More precisely, we want to ensure for some $\theta \in (0,1)$ that the Taylor polynomial satisfies $\| \epec - T_N(\mP)\ve_c \|_1 \leq \theta \eps$ and, additionally, our computed  approximation $\vx$ satisfies $\| T_N(\mP)\ve_c - \vx \|_1 \leq (1-\theta) \eps$. We pick $\theta = 1/2$, although we suspect there is an opportunity to optimize this term.

\subsection{The Gauss-Southwell Coordinate Relaxation Method}\label{sec:gs}
One of the algorithmic procedures we employ is to solve a linear system via Gauss-Southwell. The Gauss-Southwell (GS) method is an iterative method related to the Gauss-Seidel and coordinate descent methods~\cite{Luo1992-coordinate-descent}. In solving a linear system $\mA\vx = \vb$ with current solution $\vxj{x}{k}$ and residual $\vxj{r}{k} = \vb - \mA \vxj{x}{k}$, the
GS iteration acts via coordinate relaxation on the largest magnitude entry of the residual at each step, whereas the Gauss-Seidel method repeatedly cycles through all elements of the residual.
Like Gauss-Seidel, the GS method converges on diagonally dominant matrices, symmetric positive definite matrices, and $M$-matrices. It is strikingly effective when the underlying system is sparse and the solution vector can be approximated locally. Because of this, the algorithm has been reinvented in the context of local PageRank computations~\cite{andersen2006-local,berkhin2007-bookmark,jeh2003-personalized}. Next we present the basic iteration of GS.

Given a linear system $\mA \vx = \vb$ with initial solution $\vxj{x}{0} = 0$ and residual $\vxj{r}{0} = \vb$, GS proceeds as follows. To update from step $k$ to step $k+1$, set $m^{(k)}$ to be the maximum magnitude entry of $\vxj{r}{k}$, i.e. $m^{(k)} := (\vxj{r}{k})_{i_k}$; then, update the solution and residual:
\begin{equation}\label{eqn:gs}
\begin{aligned}
\vxj{x}{k+1} &= \vxj{x}{k} + m^{(k)} \cdot \ve_{i_k} & \text{ \emph{update the $i_k$th coordinate only} }\\
\vxj{r}{k+1} &= \vxj{r}{k} - m^{(k)} \cdot \mA \ve_{i_k} & \text{ \emph{update the residual}. }
\end{aligned}
\end{equation}
Observe that updating the residual $\vxj{r}{k}$ in \eqref{eqn:gs} involves adding only a scalar multiple of a column of $\mA$ to $\vxj{r}{k}$.
If $\mA$ is sparse, then the whole step involves updating a single entry of the solution $\vxj{x}{k}$, and only a small number of entries of $\vxj{r}{k}$. When $\mA = \mP$, the column-stochastic transition matrix, then updating the residual involves accessing the out-links of a single node. 

The reason that Gauss-Southwell is called a ``coordinate relaxation'' method is that it can be derived by \emph{relaxing} or \emph{freeing} the $i_k$th coordinate to satisfy the linear equations in that coordinate only. For instance, suppose for the sake of simplicity that $\mA$ has 1s on its diagonal and let $\va_{i_k}^T$ be the $i_k$th row of $\mA$. Then at the $k$th step, we choose $\vx\itn{k+1}$ such that $\va_{i_k}^T \vx\itn{k+1} = b_{i_k}$, but we allow only $x_{i_k}$ to vary -- it was the coordinate that was relaxed. Because $\mA$ has 1s on its diagonal, we can write this as: \[ x_{i_k}\itn{k+1} = b_{i_k} - \sum_{j \not= i_{k}} A_{i_k,j} x_{j}\itn{k} = (\vxj{r}{k})_{i_k} + x_{i_k}\itn{k}. \]
This is exactly the same update as in \eqref{eqn:gs}. It's also the same update as in the Gauss-Seidel method. The difference with Gauss-Seidel, as it is typically explained, is that it does not maintain an explicit residual and it chooses coordinates cyclically. 

%% file: sec-algorithms.tex
\section{Algorithms}\label{sec:alg}
We now present three algorithms for approximating $\epec$ designed for matrices from sparse networks satisfying $\| \mP \|_1 \leq 1$.
Two of the methods consist of coordinate relaxation steps on a linear system, $\mM$, that we construct from a Taylor polynomial approximating $\expof{\mP}$, as explained in Section \ref{sec:linsys}.
The first algorithm, which we call \gexpm, applies Gauss-Southwell to $\mM$ with sparse iteration vectors $\vx$ and $\vr$, and tracks elements of the residual in a heap to enable fast access to the largest entry of the residual.
The second algorithm is a close relative of \gexpm, but it stores \emph{significant entries} of the residual in a queue rather than maintaining a heap. This makes it faster, and also turns out to be closely related to a truncated Gauss-Seidel method. Because of the queue, we call this second method \gexpmq.
The bulk of our analysis in Section~\ref{sec:anl} studies how these methods converge to an accurate solution.

The third algorithm approximates the product $T_N(\mP)\ve_c$ using Horner's rule on the polynomial $T_N(\mP)$, in concert with a procedure we call an ``incomplete" matrix-vector product (Section~\ref{sec:expmimv}). This procedure  deletes all but the largest entries in the vector before performing a matrix-vector product.

We construct \gexpm and \gexpmq such that the solutions they produce have guaranteed accuracy, as proved in Section \ref{sec:anl}. On the other hand, \expmimv sacrifices predictable accuracy for a guaranteed fast runtime bound.

\subsection{Forming a Linear System}\label{sec:linsys}
We stated a coordinate relaxation method on a linear system. Thus, to use it, we require a linear system whose solution is an approximation of $\epec$.
Here we derive such a system using a Taylor polynomial for the matrix exponential. We present the construction for a general matrix $\mA \in \mathbb{R}^{n\times n}$ because the Taylor polynomial, linear system, and iterative updates are all well-defined for any real square matrix $\mA$; it is only the convergence results that require the additional assumption that $\mA$ is a graph-related matrix $\mP$ satisfying $\|\mP\|_1 \leq 1$.

Consider the product of the degree $N$ Taylor polynomial with $\ve_c$:
\[
T_N(\mA) \ve_c = \sum_{j=0}^N \tfrac{1}{j!}\mA^j \ve_c \approx \expof{\mA}\ve_c
\]
and denote the $j$th term of the sum by $\vv_j := \mA^j \ve_c /j!$. Then $\vv_0 = \ve_c$, and the later terms satisfy the recursive relation $\vv_{j+1} = \mA\vv_j /(j+1)$ for $j = 0, ... , N-1.$
This recurrence implies that the vectors $\vv_j$ satisfy the system
\begin{equation}\label{linsys}
\left[\begin{array}{ccccc}
  \mI & & & & \\
 -\mA/1 & \mI & & & \\
        & -\mA/2 & \ddots & & \\
      & & \ddots& \mI&  \\
      &    &    & -\mA/N & \mI \\
 \end{array}\right]
 \blvec{\vv} =
\bmat{\ve_c \\ 0 \\ \vdots \\ \vdots \\ 0} .
\end{equation}
If $\hat{\vv} = [\hat{\vv}_0, ... , \hat{\vv}_N]^T$ is an approximate solution to equation $\ref{linsys}$, then we have $\hat{\vv}_j \approx \vv_j$ for each term, and so $\sum_{j=0}^N \hat{\vv}_j \approx \sum_{j=0}^N \vv_j = T_N(\mA)\ve_c$. Hence, an approximate solution of this linear system yields an approximation of $\expof{\mA}\ve_c$. Because the end-goal is computing $\vx := \sum_{j=0}^N \hvv_j$, we need not form the blocks $\hvv_j$; instead, all updates that would be made to a block of $\hvv$ are instead made directly to $\vx$.

We denote the block matrix by $\mM$ for convenience; note that the explicit matrix can be expressed more compactly as $(\mI_{N+1} \kron \mI_n - \mS \kron \mA )$, where $\mS$ denotes the $(N+1) \times (N+1)$ matrix with first sub-diagonal equal to $[1/1, 1/2, ... , 1/N]$, and $\mI_k$ denotes the $k\times k $ identity matrix. Additionally, the right-hand side $[\ve_c, 0 , ... , 0]^T$ equals $\ve_1 \kron \ve_c$. When we apply an iterative method to this system, we often consider sections of the matrix $\mM = (\mI \kron \mI - \mS \kron \mA)$, solution $\hvv = [\hvv_0, ... , \hvv_N]^T$, and residual $\vr = [\vr_0, ... , \vr_N]^T$ partitioned into blocks. These vectors each consist of $N+1$ blocks of length $n$, while $\mM$ is an $(N+1)\times (N+1)$ block matrix, with blocks of size $n \times n$.

In practice, this large linear system is never formed, and we work with it implicitly. That is, when the algorithms \gexpm and \gexpmq apply coordinate relaxation to the linear system \eqref{linsys}, we will restate the iterative updates of each linear solver in terms of these blocks. We describe how this can be done efficiently for each algorithm below.

We summarize the notation introduced thus far that we will use throughout the rest of the discussion in Table~\ref{tab:notation}.

\begin{table}
\centering
\caption{Notation for adapting GS} \label{tab:notation}
\begin{tabularx}{\linewidth}{rl} \toprule
$\vx$  &  our approximation of $\epec$\\
$T_N(\mP)$  & the degree $N$ Taylor approximation to $\expof{\mP}$ \\
$\vv_j$ & term $j$ in the sum $\sum_{k=0}^N \mP^k \ve_c /k!$ \\
$\vv$ & the vector $[\vv_0, ... , \vv_N]^T$ \\
$\mM$ & the $(N+1)n \times (N+1)n$ matrix $\mI \kron \mI - \mS \kron \mP$ \\
$\mS$ & the $(N+1) \times (N+1)$ matrix with first subdiagonal $[1/1, ... , 1/N]$\\
$\hvxj{v}{k}$ & our GS approximate solution for $\mM\vv = \eoe{1}{c}$ at step $k$ \\
$\vxj{r}{k}$ & our GS residual for $\mM\vv = \eoe{1}{c}$ at step $k$ \\
$\hvxj{v}{k}_j$ & block $j$ in $\hvxj{v}{k} = [\hvxj{v}{k}_0, ... , \hvxj{v}{k}_N]^T$ \\
$(\hvxj{v}{k})_q$ & entry $q$ of the full vector $\hvxj{v}{k}$\\
IMV & an ``incomplete" matrix-vector product (Section~\ref{sec:expmimv})\\
 \bottomrule
\end{tabularx}
\end{table}

\subsection{Weighting the Residual Blocks}\label{sec:alg:psi}
Before presenting the algorithms, it is necessary to develop some understanding of the error introduced using the linear system in~\eqref{linsys} approximately. Our goal is to show that the error vector arising from using this system's solution to approximate $T_N(\mP)$ is a weighted sum of the residual blocks $\vr_j$. This is important here because then we can use the  coefficients of $\vr_j$ to determine the terminating criterion in the algorithms. To begin our error analysis, we look at the inverse of the matrix $\mM$.
\begin{lemma}\label{lem:Minverse}
Let $\mM = (\mI_{N+1} \kron \mI_n - \mS \kron \mA )$, where $\mS$ denotes the $(N+1) \times (N+1)$ matrix with first sub-diagonal equal to $[1/1, 1/2, ... , 1/N]$, and $\mI_k$ denotes the $k\times k $ identity matrix. Then
$\mM\inv = \sum_{k=0}^{N} \mS^k \kron \mA^k.$
\end{lemma}

For a proof, see Appendix~\ref{sec:app}. Next we use the inverse of $\mM$ to define our error vector in terms of the residual blocks from the linear system in Section~\ref{sec:linsys}. In order to do so, we need to define a family of polynomials associated with the degree $N$ Taylor polynomial for $e^x$:
\begin{equation}\label{psi}
\psi_j(x) := \sum_{m=0}^{N-j} \frac{j!}{(j+m)!} x^m
\end{equation}
for $j = 0, 1, ... , N$. Note that these are merely slightly altered truncations of the well-studied functions $\phi_j(x) = \sum_{m=0}^{\infty} \tfrac{x^m}{(m+j)!}$ that arise in exponential integrators, a class of methods for solving initial value problems. These polynomials $\psi_j(x)$ enable us to derive a precise relationship between the error of the polynomial approximation and the residual blocks of the linear system $\mM \vv = \ve_1 \kron \ve_c$ as expressed in the following lemma.
\begin{lemma}\label{lem:linsyserror}
Consider an approximate solution $\hat{\vv} = [ \hvv_0; \hvv_1; \cdots; \hvv_N] $ to the linear system 
\[(\mI_{N+1} \kron \mI_n - \mS \kron \mA )[ \vv_0; \vv_1; \cdots; \vv_N] = \ve_1 \kron \ve_c.
\] 
Let $\vx = \sum_{j=0}^N \hat{\vv}_j$, let $T_N(x)$ be the degree $N$ Taylor polynomial for $e^x$, and define $\psi_j(x) = \sum_{m=0}^{N-j} \tfrac{j!}{(j+m)!} x^m$. Define the residual vector $\vr = [\vr_0 ; \vr_1 ; \ldots ; \vr_N]$ by $\vr: = \ve_1 \kron \ve_c - (\mI_{N+1} \kron \mI_n - \mS \kron \mA ) \hat{\vv}$. Then the error vector $T_N(\mA)\ve_c - \vx$ can be expressed
\[
T_N(\mA)\ve_c - \vx = \sum_{j=0}^N \psi_j(\mA)\vr_j.
\]
\end{lemma}

See Appendix~\ref{sec:app} for a proof. The essence of the proof is that, using Lemma~\ref{lem:Minverse}, we can write a simple formulation for $\mM^{-1} \vr$, which is the expression for the error.

\subsection{Approximating the Taylor Polynomial via Gauss-Southwell}\label{sec:gexpm}
The main idea of \gexpm, our first algorithm, is to apply Gauss-Southwell to the system \eqref{linsys} in a way that exploits the sparsity of both $\mM$ and the input matrix $\mP$. In particular, we need to adapt the coordinate and residual updates of Gauss-Southwell in~\eqref{eqn:gs} for the system \eqref{linsys} by taking advantage of the block structure of the system.

We begin our iteration to solve  $\mM \vv = \eoe{1}{c}$ with $\hvxj{v}{0} = 0$ and $\vxj{r}{0} = \eoe{1}{c}$. Consider an approximate solution after $k$ steps of Gauss-Southwell,  $\hvxj{v}{k}$, and residual $\vxj{r}{k}$. The standard GS iteration consists of adding the largest entry of $\vxj{r}{k}$, call it $m^{(k)} := \vxj{r}{k}_q$, to $\hvxj{v}{k}_q$, and then updating $\vxj{r}{k+1} = \vxj{r}{k} - m^{(k)} \mM \ve_q$.

We want to rephrase the iteration using the block structure of our system. We will denote the $j$th block of $\vr$ by $\vr_{j-1}$, and entry $q$ of $\vr$ by $(\vr)_q$. Note that the entry $q$ corresponds with node $i$ in block $j-1$ of the residual, $\vr_{j-1}$. Thus, if the largest entry is $(\vxj{r}{k})_q$, then we write $\ve_q = \eoe{j}{i}$ and the largest entry in the residual is $m^{(k)}:= (\eoe{j}{i})^T \vxj{r}{k} = \ve_i^T \vr_{j-1}^{(k)}$.
The standard GS update to the solution would then add $m^{(k)}(\eoe{j}{i})$ to the iterative solution, $\hvxj{v}{k}$; but this simplifies to adding $m^{(k)}\ve_i$ to block $j-1$ of $\hvxj{v}{k}$, i.e. $\hvxj{v}{k}_{j-1}$. In practice we never form the blocks of $\hvv$, and instead simply add $m^{(k)}\ve_i$ to $\vxj{x}{k}$, our iterative approximation of $\epec$.

The standard update to the residual is $\vxj{r}{k+1} = \vxj{r}{k} - m^{(k)} \mM \ve_q$. Using the block notation and expanding $\mM = \mI \kron \mI - \mS \kron \mP$, the residual update becomes $\vxj{r}{k+1} = \vxj{r}{k} - m^{(k)}\eoe{j}{i} + (\mS\ve_j)\kron(\mP\ve_i).$ Furthermore, we can simplify the product $\mS\ve_j$ using the structure of $\mS$: for $j= 1, ..., N$, we have $\mS\ve_j = \ve_{j+1}/j$; if $j = N+1$, then $\mS\ve_j = 0$.

To implement this iteration, we needed $q$, the index of the largest entry of the residual vector. To ensure this operation is fast, we store the residual vector's non-zero entries in a heap. This allows $O(1)$ lookup time for the largest magnitude entry each step at the cost of reheaping the residual each time an entry of $\vr$ is altered.

We want the algorithm to terminate once its 1-norm error is below a prescribed tolerance, $\eps$. To ensure this, we maintain a weighted sum of the 1-norms of the residual blocks, $t^{(k)} = \sum_{j=0}^N \psi_j(1)\|\vxj{r}{k}_j\|_1$. Now we can reduce the entire \gexpm iteration to the following:
\begin{enumerate}[noitemsep]
\item Set $m^{(k)} = (\eoe{j}{i})^T\vxj{r}{k}$, the top entry of the heap, then delete the entry in $\vxj{r}{k}$ so that $(\eoe{j}{i})^T\vxj{r}{k+1} = 0$.
\item Update $\vxj{x}{k+1} = \vxj{x}{k} + m^{(k)}\ve_i$. 
\item If $j < N+1$, update $\vxj{r}{k+1}_{j} = \vxj{r}{k}_{j} + m^{(k)}\mP\ve_i/j$, reheaping $\vr$ after each add.
\item Update $t^{(k+1)} = t^{(k)} - \psi_{j-1}(1)|m^{(k)}| + \psi_j(1)|m^{(k)}|/j$.
\end{enumerate}

We show in Theorem~\ref{thm:conv:gexpm} that iterating until $t^{(k)} \leq \eps$ guarantees a 1-norm accuracy of $\eps$. We discuss the complexity of the algorithms in Section~\ref{sec:anl}.

\subsection{Approximating the Taylor Polynomial via Gauss-Seidel}\label{sec:gexpmq}
Next we describe a similar algorithm that stores the residual in a queue to avoid the heap updates. Our original inspiration for this method was the relationship between the Bookmark Coloring Algorithm~\cite{berkhin2007-bookmark} and the Push method for Personalized PageRank~\cite{andersen2006-local}. The rationale for this change is that maintaining the heap in \gexpm is slow. Remarkably, the final algorithm we create is actually a version of Gauss-Seidel that \emph{skips} updates from insignificant residuals, whereas standard  Gauss-Seidel cycles through coordinates of the matrix cyclically in index order. Our algorithm will use the queue to do \emph{one} such pass and maintain \emph{significant entries} of the residual that must be relaxed (and not skipped). 

The basic iterative step is the same as in \gexpm,  except that the entry of the residual chosen, say $(\eoe{j}{i})^T\vr$, is not selected to be the largest in $\vr$. Instead, it is the next entry in a queue storing significant entries of the residual. Then as entries in $\vr$ are updated, we place them at the back of the queue, $Q$. Note that the block-wise nature of our update has the following property: an update from the $j$th block results in residuals changing in the $(j+1)$st block. Because new elements are added to the tail of the queue, all entries of $\vr_{j-1}$ are relaxed before proceeding to $\vr_{j}$.

If carried out exactly as described, this would be equivalent to performing each product $\vv_j = \mP \vv_{j-1}/j$ in its entirety. But we want to avoid these full products; so we introduce a rounding threshold for determining whether or not to operate on the entries of the residual as we pop them off of $Q$.

The rounding threshold is determined as follows. After every entry in $\vr_{j-1}$ is removed from the top of $Q$, then all entries remaining in $Q$ are in block $\vr_{j}$ (remember, this is because operating on entries in $\vr_{j-1}$ adds to $Q$ only entries that are from $\vr_{j}$.) Once every entry in $\vr_{j-1}$ is removed from $Q$, we set $Z_{j} = |Q|$, the number of entries in $Q$; this is equivalent to the total number of non-zero entries in $\vr_{j}$ before we begin operating on entries of $\vr_{j}$. Then, while operating on $\vr_{j}$, the threshold used is
\begin{equation}\label{eqn:threshold}
\textrm{threshold}(\eps,j,N) =
\tfrac{\eps}{N\psi_j(1)Z_j}.
\end{equation}
Then, each step, an entry is popped off of $Q$, and if it is larger than this threshold, it is operated on; otherwise, it is simply discarded, and the next entry of $Q$ is considered.
Once again, we maintain a weighted sum of the 1-norms of the residual blocks, $t^{(k)} = \sum_{j=0}^N \psi_j(1)\|\vxj{r}{k}_j\|_1$, and terminate once $t^{(k)} \leq \eps$, or if the queue is empty.

Step $k+1$ of \gexpmq is as follows:
\begin{enumerate}[noitemsep]
\item Pop the top entry of $Q$, call it $r = (\eoe{j}{i})^T\vxj{r}{k}$, then delete the entry in  $\vxj{r}{k}$, so that $(\eoe{j}{i})^T\vr^{(k+1)}=0$.
\item If $r \geq \text{threshold}(\eps,j,N)$ do the following:
	\begin{enumerate}[noitemsep,nolistsep]
	\item Add $r \ve_i$ to $\vx_i$.
	\item Add $r \mP\ve_i/j$ to residual block $\vxj{r}{k+1}_{j}$.
	\item For each entry of $\vxj{r}{k+1}_{j}$ that was updated, add that entry to the back of $Q$.
	\item Update $t^{(k+1)} = t^{(k)} - \psi_{j-1}(1)|r| + \psi_j(1)|r|/j$.
	\end{enumerate}
\end{enumerate}
We also provide a working python pseudocode for this method in Figure~\ref{alg:gexpmq}. 

We show in the proof of Theorem~\ref{thm:conv:gexpmq} that iterating until $t^{(k)} \leq \eps$, or until all entries in the queue satisfying the threshold condition have been removed, will guarantee that the resulting vector $\vx$ will approximate $\epec$ with the desired accuracy. 

\begin{figure}
\setlength{\columnsep}{5pt}
\begin{multicols}{2}
\begin{lstlisting}[language=python,commentstyle={\color{blue}\itshape},xleftmargin=0pt]
## Estimate column c of 
## the matrix exponential vector 
# G is the graph as a dictionary-of-sets, 
# eps is set to stopping tolerance
def compute_psis(N):
  psis = {}
  psis[N] = 1.
  for i in xrange(N-1,-1,-1):
    psis[i] = psis[i+1]/(float(i+1.))+1.
  return psis    
def compute_threshs(eps, N, psis):
  threshs = {}
  threshs[0] = (math.exp(1)*eps/float(N))/psis[0]
  for j in xrange(1, N+1):
    threshs[j] = threshs[j-1]*psis[j-1]/psis[j]
  return threshs
## Setup parameters and constants
N = 6  
c = 1 # the column to compute
psis = compute_psis(N)
threshs = compute_threshs(eps,N,psis)
## Initialize variables
x = {} # Store x, r as dictionaries
r = {} # initialize residual
Q = collections.deque() # initialize queue
sumresid = 0.    
r[(c,0)] = 1.
Q.append(c)
sumresid += psis[0]

## Main loop
for j in xrange(0, N):
  qsize = len(Q)
  relaxtol = threshs[j]/float(qsize)
  for qi in xrange(0, qsize):
    i = Q.popleft()
    rij = r[(i,j)]
    if rij < relaxtol:
      continue
    # perform the relax step
    if i not in x: x[i] = 0.
    x[i] += rij
    r[(i,j)] = 0.
    sumresid -= rij*psis[j]
    update = (rij/(float(j)+1.))/len(G[i])
    for u in G[i]:   # for neighbors of i
      next = (u, j+1)
      if j == N-1: 
          if u not in x: x[u] = 0.
          x[u] += update
      else:
          if next not in r: 
            r[next] = 0.
            Q.append(u)
          r[next] += update
          sumresid += update*psis[j+1]
    # after all neighbors u
    if sumresid < eps: break
  if len(Q) == 0: break
  if sumresid < eps: break
\end{lstlisting}        
\end{multicols}
\caption{A working python code to implement the \gexpmq algorithm with a queue. A full demo is available from \url{https://gist.github.com/dgleich/10224374}. We implicitly normalize the graph structure into a stochastic matrix by dividing by the degree in the computation of \texttt{update}.}
\label{alg:gexpmq}
\end{figure}

\subsection{A sparse, heuristic approximation}\label{sec:expmimv}
The above algorithms guarantee that the final approximation attains the desired accuracy $\eps$. Here we present an algorithm designed to be faster. 
Because we have no error analysis for this algorithm currently, and because the steps of the method are well-defined for any $\mA \in \mathbb{R}^{n \times n}$, we discuss this algorithm in a more general setting. This method also uses a Taylor polynomial for $\expof{\mA}$, but does not use the linear system constructed for the previous two methods. Instead, the Taylor terms are computed via Horner's rule on the Taylor polynomial. But, rather than a full matrix-vector product, we apply what we call an``incomplete'' matrix-vector product (IMV) to compute the successive terms. Thus, our name: \expmimv. We describe the IMV procedure before describing the algorithm.

\paragraph{Incomplete Matrix-vector Products (IMV)}
Given any matrix $\mA$ and a vector $\vv$ of compatible dimension, the IMV procedure sorts the entries of $\vv$, then removes all entries except for the largest $z$. Let $[\vv]_z$ denote the vector $\vv$ with all but its $z$ largest-magnitude entries deleted. Then we define the \textit{$z$-incomplete matrix-vector product} of $\mA$ and $\vv$ to be $\mA [\vv]_z$. We call this an \textit{incomplete} product, rather than a rounded matrix-vector product, because, although the procedure is equivalent to rounding to 0 all entries in $\vv$ below some threshold, that rounding-threshold is not known a priori, and its value will vary from step to step in our algorithm.

There are likely to be a variety of ways to implement these IMVs. Ours computes $[\vv]_z$ by filtering all the entries of $\vv$ through a min-heap of size $z$. For each entry of $\vv$, if that entry is larger than the minimum value in the heap, then replace the old minimum value with the new entry and re-heap; otherwise, set that entry in $[\vv]_z$ to be zero, then proceed to the next entry of $\vv$. Many similar methods have been explored in the literature before, for instance~\cite{yuan-2011-truncated}.

\paragraph{Horner's rule with IMV}
A Horner's rule approach to computing $T_N(\mA)$ considers the polynomial as follows:
\begin{equation}
\begin{aligned}
\expof{\mA} \approx T_N(\mA) & = \mI + \tfrac{1}{1!}\mA^1 + \tfrac{1}{2!}\mA^2 + \cdots + \tfrac{1}{N!}\mA^N \\
&= \mI + \tfrac{1}{1}\mA\left(\mI + \tfrac{1}{2}\mA\left(\mI + \cdots + \tfrac{1}{N-1}\mA\left(\mI + \tfrac{1}{N}\mA\right) \cdots \right)\right)\\
\end{aligned}
\end{equation}
Using this representation, we can approximate $\expof{\mA}\ve_c$ by multiplying $\ve_c$ by the inner-most term, $\mA/N$, and working from the inside out. More precisely, the \expmimv procedure is as follows:
\begin{enumerate}[noitemsep]
\item Fix $z \in \mathbb{N}$.
\item Set $\vxj{x}{0} = \ve_c$.
\item For $k = 0, ... , N-1$ compute $\displaystyle \vxj{x}{k+1} = \mA \left([\vxj{x}{k}]_z/(N-k)\right) + \ve_c$. 
\end{enumerate}
Then at the end of this process we have $\vxj{x}{N} \approx T_N(\mA)\ve_c$. The vector $[\vxj{x}{k}]_z$ used in each iteration of step 3 is computed via the IMV procedure described above. For an experimental analysis of the speed and accuracy of \expmimv, see Section~\ref{sec:exp:expmimv}. 

\paragraph{Runtime analysis}
Now assume that the matrix $\mA$ in the above presentation corresponds to a graph, and let $d$ be the maximum degree found in the graph related to $\mA$.
Each step of \expmimv requires identifying the $z$ largest entries of $\vxj{v}{k}$, multiplying $\mA[\vxj{v}{k}]_z$, then adding $\ve_c$. If $\vv$ has $\nnz{\vv}$ non-zeros, and the largest $z$ entries are desired, then computing $[\vv]_z$ requires at most $O(\nnz{\vv} \log(z))$  work: each of the $\nnz{\vv}$ entries are put into the size-$z$ heap, and each heap update takes at most $O(\log(z))$ operations.

Note that the number of non-zeros in $\vxj{v}{k}$, for any $k$, can be no more than $dz$. This is because the product
$\vxj{v}{k} = \mA[\vxj{v}{k-1}]_z$
combines exactly $z$ columns of the matrix: the $z$ columns corresponding to the $z$ non-zeros in $[\vxj{v}{k}]_z$. Since no column of $\mA$ has more than $d$ non-zeros, the sum of these $z$ columns can have no more than $dz$ non-zeros.
Hence, computing $[\vxj{v}{k}]_z$ from $\vxj{v}{k}$ requires at most $O(dz\log(z))$ work. Observe also that the work done in computing the product $\mA[\vxj{v}{k}]_z$ cannot exceed $dz$.
Since exactly $N$ iterations suffice to evaluate the polynomial, we have proved Theorem~\ref{thm:expmimv}:
\begin{theorem}\label{thm:expmimv}
Let $\mA$ be any graph-related matrix having maximum degree $d$. Then the \expmimv procedure, using a heap of size $z$, computes an approximation of $\expof{\mA}\ve_c$ via an $N$ degree Taylor polynomial in work bounded by $O(Ndz\log z)$.
\end{theorem}

If $\mA$ satisfies $\|\mA\|_1\leq 1$, then by Lemma~\ref{lem:Nbound} we can choose $N$ to be a small constant to achieve a coarse $O(10^{-3})$ approximation. 

While the \expmimv method always has a sublinear runtime, we currently have no theoretical analysis of its accuracy. However, in our experiments we found that a heap size of $z = 10,000$ yields a 1-norm accuracy of $\approx 10^{-3}$ for social networks with millions of nodes (Section~\ref{sec:exp:expmimv}). Yet, even for a fixed value of $z$, the accuracy varied widely. For general-purpose computation of the matrix exponential, we do not recommend this procedure.
If instead the purpose is identifying large entries of $\epec$, our experiments suggest that \expmimv often accomplishes this task with high accuracy (Section~\ref{sec:exp:expmimv}).

%% file: sec-analysis-dgleich.tex
\section{Analysis}\label{sec:anl}

We divide our theoretical analysis into two stages. In the first we establish the convergence of the coordinate relaxation methods, \gexpm and \gexpmq, for a class of matrices that includes column-stochastic matrices. Then, in Section~\ref{sec:anl:pow}, we give improved results when the underlying graph has a degree distribution that follows a power-law, which we define formally in Section~\ref{sec:anl:pow}.

\subsection{Convergence of Coordinate Relaxation Methods}
In this section, we show that both \gexpm and \gexpmq converge to an approximate solution with a prescribed 1-norm error $\eps$ for any matrix $\mP$ satisfying $\|\mP\|_1 \leq 1$. 

Consider the large linear system \eqref{linsys} using a matrix $\mP$ with 1-norm bounded by one. Then by applying both Lemma~\ref{lem:linsyserror} and the triangle inequality, we find that the error in approximately solving the system can be expressed in terms of the residuals in each block: 
\[ \|T_N(\mP)\ve_c - \vx\|_1 \leq \sum_{j=0}^N \| \psi_j(\mP) \|_1 \|\vr_j\|_1. \]
Because the polynomials $\psi_j(t)$ have all nonnegative coefficients, and because $\psi_j(\mP)$ is a polynomial in $\mP$ for each $j$, we have that $\|\psi_j(\mP) \|_1 \leq \psi_j(\|\mP\|_1 )$. Finally, using the condition that $\|\mP\|_1 \leq 1$, we have proved the following:
\begin{lemma}\label{lem:errorpsi}
Consider the setting from Lemma~\ref{lem:linsyserror} applied to a matrix $\|\mP\|_1 \leq 1$. Then the norm of the error vector $T_N(\mA)\ve_c - \vx$ associated with an approximate solution is a weighted sum of the residual norms from each block:
\[
\|T_N(\mP)\ve_c - \vx\|_1 \leq \sum_{j=0}^N \psi_j(1)\|\vr_j\|_1 .
\]
\end{lemma}
Note that this does not require nonnegativity of either $\vr$ or $\mP$, only that $\| \mP \|_1 \leq 1$; this improves on our original analysis in~\cite{Kloster-2013-nexpokit}. 

We now show that our algorithms monotonically decrease the weighted sum of residual norms, and hence converge to a solution. The analysis differs between the two algorithms (Theorem~\ref{thm:conv:gexpm} for \gexpm and Theorem~\ref{thm:conv:gexpmq} for \gexpmq), but the intuition remains the same: each relaxation step reduces the residual in block $j$ and increases the residual in block $j+1$, but by a smaller amount. Thus, the relaxation steps monotonically reduce the residuals. 


\begin{theorem}\label{thm:conv:gexpm}
Let $\mP \in \mathbb{R}^{n \times n}$ satisfy $\|\mP \|_1 \leq 1$. Then in the notation of Section~\ref{sec:gexpm}, the residual vector after $l$ steps of \gexpm satisfies
$\| \vr^{(l)} \|_1 \leq l^{- 1/(2d)}$
and the error vector satisfies
\begin{equation}\label{eqn:thm:gexpmconv}
\| T_N(\mP)\ve_c - \vx \|_1 \leq \exp(1) \cdot l^{\lrp{ - \tfrac{1}{2d} }},
\end{equation}
so \gexpm converges in at most $l = (\exp(1)/\eps)^{2d}$ iterations.
\end{theorem}

\begin{proof}
The iterative update described in Section~\ref{sec:gexpm} involves a residual block, say $\vr_{j-1}$, and a row index, say $i$, so that the largest entry in the residual at step $l$ is $m^{(l)} = (\eoe{j}{i})^T\vxj{r}{l}$. First, the residual is updated by deleting the value $m^{(l)}$ from entry $i$ of the block $\vxj{r}{l}_{j-1}$, which results in the 1-norm of the residual decreasing by exactly $|m^{(l)}|$. Then, we
add $m^{(l)} \mP \ve_i / j$ to $\vxj{r}{l}_j$, which results in the 1-norm of the residual increasing by at most $\|m^{(l)} \mP \ve_i /j\|_1 \leq |m^{(l)}/j|$, since $\|\mP\ve_i\|_1 \leq 1$. Thus, the net change in the 1-norm of the residual will satisfy
\[
\|\vxj{r}{l+1}\|_1 \leq \|\vxj{r}{l}\| - |m^{(l)}| + \left|\tfrac{m^{(l)}}{j}\right|.
\]

Note that the first residual block, $\vr_0$, has only a single non-zero in it, since $\vr_0 = \ve_c$ in the initial residual. This means that every step after the first operates on residual $\vr_{j-1}$ for $j \geq 2$. Thus, for every step after step 0, we have that $1/j \leq 1/2$. Hence, we have
\[
\|\vxj{r}{l+1}\|_1 \leq \|\vxj{r}{l}\| - |m^{(l)}| + \left|\tfrac{m^{(l)}}{2}\right| = \|\vxj{r}{l}\| -  \tfrac{|m^{(l)}|}{2}.
\]
We can lowerbound $|m^{(l)}|$, the largest-magnitude entry in the residual, with the average magnitude of the residual. The average value of $\vr$ equals $\|\vr\|_1$ divided by the number of non-zeros in $\vr$.
After $l$ steps, the residual can have no more than $dl$ non-zero elements, since at most $d$ non-zeros can be introduced in the residual each time $\mP\ve_i$ is added; hence, the average value at step $l$ is lowerbounded by $\|\vxj{r}{l}\|_1/dl$. Substituting this into the previous inequality, we have
\[
\|\vxj{r}{l+1}\|_1 \leq \|\vxj{r}{l}\| -  \tfrac{|m^{(l)}|}{2}
\leq \|\vxj{r}{l}\| -  \tfrac{\|\vxj{r}{l}\|_1}{2dl}
= \|\vxj{r}{l}\| \lrp{ 1 - \tfrac{1}{2dl} }.
\]
Iterating this inequality yields the bound 
$\|\vr^{(l)}\|_1 \leq  \|\vr^{(0)}\|_1 \prod_{k=1}^l ( 1 - 1/(2dk) )
$,
and since $\vr^{(0)} = \eoe{1}{c}$ we have $\|\vr^{(0)}\|_1 = 1$. Thus, $\|\vr^{(l)}\|_1 \leq  \prod_{k=1}^l ( 1 - 1/(2dk) )$. The first inequality of \eqref{thm:conv:gexpm} follows from using the facts $(1+x) \leq e^{x}$ (for $x > -1$) and $\log(l) < \sum_{k=1}^l1/k$ to write
\[
 \prod_{k=1}^l ( 1 - \tfrac{1}{2dk} )  \leq  \exp\biggl\{- \tfrac{1}{2d}\sum_{k=1}^l \tfrac{1}{k} \biggr\}  \leq  \expof{-\tfrac{1}{2d}\log l} = l^{\lrp{- \frac{1}{2d} }}.
\]
The inequality $(1+x) \leq e^{x}$ follows from the Taylor series $e^x = 1 + x + o(x^2)$, and the lowerbound for the partial harmonic sum $\sum_{k=1}^l1/k$ follows from the left-hand rule integral approximation $\log(l) = \int_1^l (1/x) \, dx < \sum_{k=1}^l1/k$.

Finally, to prove inequality~\eqref{eqn:thm:gexpmconv}, we use the fact from the proof of Lemma 3 in~\cite{Kloster-2013-nexpokit} that $\psi_j(1) \leq \psi_0(1) \leq \exp(1)$ for all $j = 0, ... , N$. For the readers' convenience, we include a proof of the inequalities $\psi_j(1) \leq \psi_0(1) \leq \exp(1)$ in the appendix. Thus, we have
$ \| T_N(\mP)\ve_c - \vx \|_1 \leq \psi_0(1) \sum_{j=0}^N \|\vr_j\|_1$ by Lemma~\ref{lem:errorpsi}. Next, note that $\sum_{j=0}^N \|\vr_j\|_1 = \|\vr\|_1$, because $\vr = [\vr_0, \vr_1, ... , \vr_N]^T$. Combining these facts we have
$\| T_N(\mP)\ve_c - \vx \|_1 \leq \exp(1) \|\vr\|_1$, which proves the error bound. The bound on the number of iterations required for convergence follows from simplifying the inequality $\exp(1)l^{-1/(2d)} < \eps$.
\end{proof}
Next we state the convergence result for \gexpmq.
\begin{theorem}\label{thm:conv:gexpmq}
Let $\mP \in \mathbb{R}^{n \times n}$ satisfy $\|\mP \|_1 \leq 1$. Then in the notation of Section~\ref{sec:gexpmq}, using a threshold of
\[
\text{threshold}(\eps,j,N) = \tfrac{\eps}{N\psi_j(1)Z_j}
\]
for each residual block $\vr_j$ will guarantee that
when \gexpmq terminates, the error vector satisfies
$\| T_N(\mP)\ve_c - \vx \|_1 \leq \eps$.
\end{theorem}
\begin{proof}
From Lemma~\ref{lem:errorpsi} we have $ \|T_N(\mP)\ve_c - \vx\|_1 \leq \sum_{j=0}^N \psi_j(1)\|\vr_j\|_1 $.
During the first iteration we remove the only non-zero entry in $\vr_0 = \ve_c$ from the queue, then add $\mP\ve_c$ to $\vr_1$. Thus, when the algorithm has terminated, we have $\|\vr_0\|_1 = 0$, and so we can ignore the term $\psi_0(1)\|\vr_0\|_1$ in the sum. In the other $N$ blocks of the residual, $\vr_j$ for $j = 1, ... , N$, the steps of \gexpmq delete every entry with magnitude $r$ satisfying
$r \geq \eps/(N\psi_j(1)Z_j)$.
This implies that all entries remaining in block $\vr_j$ are bounded above in magnitude by $\eps/(N\psi_j(1)Z_j)$. Since there can be no more than $Z_j$ non-zero entries in $\vr_j$ (by definition of $Z_j$), we have that
$\|\vr_j\|_1$ is bounded above by $Z_j \cdot \eps/(N\psi_j(1)Z_j)$. Thus, we have
\begin{align*}
\|T_N(\mP)\ve_c - \vx\|_1  &\leq \sum_{j=1}^N \psi_j(1) \lrp{Z_j \tfrac{\eps}{N\psi_j(1)Z_j} }
\end{align*}
and simplifying completes the proof.
\end{proof}

Currently we have no theoretical runtime analysis for \gexpmq. However, because of the algorithm's similarity to \gexpm, and because of our strong heuristic evidence (presented in Section~\ref{sec:exp}), we believe a rigorous theoretical runtime bound exists.

%% file: sec-powerlaw.tex
\section{Networks with a Power-Law Degree Distribution}\label{sec:anl:pow}

In our convergence analysis for \gexpm in Section~\ref{sec:anl}, the inequalities rely on our estimation of the largest entry in the residual vector at step $l$, $m^{(l)}$. In this section we achieve a tighter bound on $m^{(l)}$ by using the distribution of the degrees of the underlying graph instead of just $d$, the maximum degree. In the case that the degrees follow a power-law distribution, we show that the improvement on the bound on $m^{(l)}$ leads to a sublinear runtime for the algorithm.

The degree distribution of a graph is said to follow a power-law if the $k$th largest degree of the graph, $d(k)$, satisfies $d(k) = Q\cdot d \cdot k^{-p}$ for $d = d(1)$ the largest degree in the graph, and positive constants $Q$ and $p$. A degree distribution of this kind applies to a variety of real-world networks~\cite{Faloutsos-1999-power-law}. A more commonly-used definition states that the number of nodes having degree $k$ is equal to $k^{-a}$, but the two definitions can be shown to be equivalent for a certain range of values of their respective exponents, $p$ and $a$~\cite{adamic-zipf-power-law}. In this definition, the values of the exponent $a$ for real-world networks range from 2 to 3, frequently closer to 2. These values correspond to $p=1$ (for $a=2$) and $p=1/2$ (for $a = 3$) in the definition that we use. Finally, we note that, though the definition that we use contains an equality, our results hold for any graph with a degree distribution satisfying a ``sub" power-law, meaning $d(k) \leq Q\cdot k^{-p}$. We now state our main result, then establish some preliminary technical lemmas before finally proving it.
\begin{theorem}\label{thm:work} For a graph with degree distribution following a power-law with $p \in (0,1]$, max degree $d$, and minimum degree $\delta$, \gexpm converges to a $1$-norm error of $\eps$ in work bounded by
\begin{equation}\label{work}
\text{work$(\eps)$} =
  \begin{cases} 
      \hfill O\left(\log(\tfrac{1}{\eps}) \left(\tfrac{1}{\eps}\right)^{\frac{3\delta}{2}} d^2\log(d) \max\{ \log(d), \log(\tfrac{1}{\eps}) \} \right) \hfill & \text{ if $p=1$ } \\
      \hfill O\left(\log(\tfrac{1}{\eps}) \left(\tfrac{1}{\eps}\right)^{\frac{3\delta}{2}} d^{1+\frac{1}{p}} \max\{ \log(d), \log(\tfrac{1}{\eps}) \} \right) \hfill & \text{ if $p\neq 1$ } \\
  \end{cases}
\end{equation}
\end{theorem}

Note that when the maximum degree satisfies $d < n^r$ for any $r < 1/(1+1/p)$, and the minimum degree is a constant independent of $n$, Theorem~\ref{thm:work} implies that the runtime scales sublinearly with the graph size, for a fixed 1-norm error of $\eps$.

In practice, having a minimum degree that is a small constant independent of $n$ is extremely common, and values of $p$ are typically near or slightly less than 1. The condition on the maximum degree (that $d < n^r$ for $r < 1/(1+1/p)$) is slightly less common, with five of our seven datasets (listed in Table \ref{tab:datasets}) satisfying $d < 2.5\cdot n^{1/2}$.

\subsection{Bounding the Number of Non-zeros in the Residual}

In the proof of Theorem $\ref{thm:conv:gexpm}$ we showed that the residual update satisfies
$
 \| \vxj{r}{l+1} \|_1 \leq \| \vxj{r}{l} \|_1 - m^{(l)}(1 - 1/j)
$,
where $m^{(l)}$ is the largest entry in $\vxj{r}{l}$, and $\vr_{j-1}$ is the section of the residual vector where the entry $m^{(l)}$ is located. We used the bound $m^{(l)} \ge \| \vxj{r}{l} \|_1 /(d l)$, which is a lowerbound on the average value of all entries in $\vxj{r}{l}$. This follows from the loose upperbound $d l$ on the number of non-zeros in $\vxj{r}{l}$. We also used the naive upperbound $1/2$ on $(1 - 1/j)$. Here we prove new bounds on these quantities. For the sake of simpler expressions in the proofs, we express the number of iterations as a multiple of $N$, i.e. $Nl$.
\begin{lemma}\label{lem:nnz} Let $d(k)$ := the $k$th largest degree in the graph (with repetition), let $f(m) := \sum_{k=1}^m d(k) $, and let $\nnz{l} :=$ the number of non-zero entries in $\vxj{r}{l}$. Then after $Nl$ iterations of \gexpm we have
\begin{equation}\label{nnz} \nnz{N l} \leq N  f(l). \end{equation}
\end{lemma}
\begin{proof}
At any given step, the number of new non-zeros we can create in the residual vector is bounded above by the largest degree of all the nodes which have not already had their neighborhoods added to $\vxj{r}{Nl}$. If we have already explored the node with degree = $d(1)$, then the next node we introduce to the residual cannot add more than $d(2)$ new non-zeros to the residual, because the locations in $\vr$ in which the node $d(1)$ would create non-zeros already have non-zero value.

We cannot conclude $\nnz{l} \leq \sum_{k=1}^l d(k) = f(l)$ because this ignores the fact that the same set of $d(1)$ nodes can be introduced into each different time step of the residual, $j = 2, \cdots, N$. Recall that entries of the residual are of the form $\ve_j \otimes \ve_i$ where $i$ is the index of the node, $i = 1, \cdots, n$; and $j$ is the section of the residual, or time step: $j = 2, \cdots, N$ (note that $j$ skips 1 because the first iteration of GS deletes the only entry in section $j=1$ of the residual).  Recall that the entry of $\vr$ corresponding to the 1 in the vector $\eoe{j}{i}$ is located in block $\vr_{j-1}$. Then for each degree, $d(1), d(2), \cdots, d(l)$, we have to add non-zeros to that set of $d(k)$ nodes in each of the $N-1$ different blocks $\vr_{j-1}$ before we move on to the next degree, $d(l+1)$:
\begin{align*}
 \nnz{Nl} &\leq d(1) + \cdots + d(1) + d(2) + \cdots + d(2) + \cdots + d(l)\\
 &\leq Nd(1) + Nd(2) + \cdots + N d(l)
 \end{align*}
which equals $N\cdot \sum_{k=1}^ld(k) = Nf(l)$.
\end{proof}

Lemma~\ref{lem:nnz} enables us to rewrite the inequality
$- m_{Nl} \leq -\|\vxj{r}{Nl}\| / (dNl)$ from the proof of Theorem~\ref{thm:conv:gexpm} as $-m_{Nl} \leq -\|\vxj{r}{Nl}\| / (N f(l))$. Letting $\sigma_{Nl}$ represent the value of $(1 - 1/j)$ in step $Nl$, we can write our upperbound on $\|\vxj{r}{Nl+1}\|_1$ as follows:
\begin{equation}\label{newbd}
 \| \vxj{r}{Nl+1} \|_1 \leq \| \vxj{r}{Nl} \|_1\left( 1 - \tfrac{\sigma_{Nl}}{N f(l)} \right).
\end{equation}

We want to recur this by substituting a similar inequality in for $\| \vxj{r}{Nl} \|_1$, but the indexing does not work out because inequality~\eqref{nnz} holds only when the number of iterations is of the form $Nl$. We can overcome this by combining $N$ iterations into one expression:

\begin{lemma}\label{lem:recur}
In the notation of Lemma~\ref{lem:nnz}, let 
$ s_{l+1} := \min\{ \sigma_{Nl + N}, \sigma_{Nl + N - 1}, \cdots, \sigma_{Nl+1} \}$.
Then the residual from $Nl$ iterations of \gexpm satisfies
 \begin{equation}\label{recur}
 \| \vxj{r}{N(l+1)} \|_1 \leq \| \vxj{r}{Nl} \|_1\left( 1 - \tfrac{s_{l+1}}{N f(l+1)} \right)^N .
\end{equation}

\end{lemma}
\begin{proof}
By Lemma~\ref{lem:nnz} we know that $Nf(l) \geq \nnz{Nl}$ for all $l$. In the proof of Lemma~\ref{lem:nnz} we showed that, during step $Nl$, no more than $d(l)$ new non-zeros can be created in the residual vector.  By the same argument, no more than $d(l+1)$ non-zeros can be created in the residual vector during steps $Nl +k$, for $k = 1, ... , N$. Thus we have $\nnz{Nl+k} \leq Nf(l) + k\cdot d(l+1) \leq Nf(l) + Nd(l+1) = Nf(l+1)$ for $k = 0, 1, ... , N$. With this we can bound $-m_{Nl+k} \leq -\|\vxj{r}{Nl+k}\|_1 / (Nf(l+1))$ for $k=0, ... , N$. Recall we defined $\sigma_{Nl}$ to be the value of $(1-1/j)$ in step $Nl$. With this in mind, we establish a new bound on the residual decrease at each step:
\begin{align*}
 \| \vxj{r}{N(l+1)} \|_1 &\leq \| \vxj{r}{Nl+N-1} \|_1 - m_{Nl+N-1}\sigma_{Nl+N-1} \\
 &\leq \| \vxj{r}{Nl+N-1} \|_1  - \tfrac{ \|\vxj{r}{Nl+N-1} \|_1 \sigma_{Nl+N-1} }{ Nf(l+1) } \\
 &= \| \vxj{r}{Nl+N-1} \|_1 \left( 1 - \tfrac{\sigma_{Nl+N-1} }{ Nf(l+1) } \right) \\
 &\leq  \left( \| \vxj{r}{Nl+N-2} \|_1 - m_{Nl+N-2}\sigma_{Nl+N-2}  \right)\left( 1 - \tfrac{\sigma_{Nl+N-1} }{ Nf(l+1) } \right) \\
&\leq  \| \vxj{r}{Nl+N-2} \|_1 \left( 1 - \tfrac{\sigma_{Nl+N-2} }{ Nf(l+1) } \right) \left( 1 - \tfrac{\sigma_{Nl+N-1} }{ Nf(l+1) } \right),
\end{align*}
and recurring this yields
\[
 \| \vxj{r}{N(l+1)} \|_1
 \leq \| \vxj{r}{Nl} \|_1  \prod_{t=1}^N\left(1 - \tfrac{\sigma_{Nl+N-t}}{Nf(l+1)}\right).
\]
 
From the definition of $s_{l+1}$ in the statement of the lemma, we can upperbound $-\sigma_{Nl+N-t}$ in the last inequality with $-s_{l+1}$. This enables us to replace the product in the last inequality with 
$\left(1 - s_{l+1} / (Nf(l+1)) \right)^N$, which proves the lemma.
\end{proof}

Recurring the inequality in \eqref{recur} bounds the residual norm in terms of $f(m)$:
\begin{corollary}\label{cor:newresbd}
In the notation of Lemma~\ref{lem:recur}, after $Nl$ iterations of \gexpm the residual satisfies
 \begin{equation}\| \vxj{r}{Nl} \|_1 \leq \expof{ - \sum_{k=1}^{l} \tfrac{s_{k}}{f(k)}  } \label{newresbd} \end{equation}
\end{corollary}
\begin{proof}
By recurring the inequality of Lemma~\ref{lem:recur}, we establish the new bound
$
 \| \vxj{r}{Nl} \|_1  \leq \| \vxj{r}{0} \|_1 \prod_{k=1}^{l}\left( 1 - s_k/(N f(k)) \right)^N  $.
The factor $\left( 1 - s_k / (N f(k)) \right)^N$ can be upperbounded by $\expof{ \sum_{k=1}^{l}\left( - s_k / (N f(k)) \right)\cdot N  }$, using the inequality $1-x\leq \exp(-x)$.
Cancelling the factors of $N$ and noting that $\|\vxj{r}{0}\|_1 = 1$ completes the proof.
\end{proof}

We want an upperbound on $- \sum_{k=1}^m 1/f(k)$, so we need a lowerbound on $\sum_{k=1}^m 1/f(k)$. This requires a lowerbound on $1/f(k)$, which in turn requires an upperbound on $f(k)$. So next we upperbound $f(k)$ using the degree distribution, which ultimately will allow us to upperbound $\|\vxj{r}{Nl}\|_1$ by an expression of $d$, the max degree.

\subsection{Power-Law Degree Distribution}\label{sec:pldd}
If we assume that the graph has a power-law degree distribution, then we can bound $d(k) \leq Q \cdot d \cdot k^{-p}$ for constants $p$ and $Q > 0$. Let $\delta$ denote the minimum degree of the graph and note that, for sparse networks in which $|E| = O(n)$, $\delta$ is a small constant (which, realistically, is $\delta=1$ in real-world networks). We can assume that $Q = 1$ because it will be absorbed into the constant in the big-O expression for work in Theorem~\ref{thm:work}. With these bounds in place, we will bound $f(k)$ in terms of $d$, $\delta$, and $p$.

\begin{lemma}\label{lem:deg} Define
\begin{equation}
C_p = 
  \begin{cases} 
      \hfill d(1+\log d) \hfill & \text{ if $p=1$ } \\
      \hfill \tfrac{1}{(1-p)}d^{\frac{1}{p}} \hfill & \text{ if $p\in (0,1)$ .} \\
  \end{cases}
\end{equation}
Then in the notation described above we have
$f(k) \leq C_p + \delta k$.
\end{lemma}
\begin{proof} The power-law bound on the degrees states that $d(k) \leq d \cdot k^{-p}$. Note that for $k > (d/\delta)^{\frac{1}{p}}$ the power-law definition of $d(k)$ above implies $d(k) < \delta$, the minimum degree, which is impossible. This leads to two cases: $k > (d/\delta)^{\frac{1}{p}}$ and $k \leq (d/\delta)^{\frac{1}{p}}$. 

The sum of the first $\floor{(d/\delta)^{\frac{1}{p}}}$ terms is $\sum_{t=1}^{\floor{(d/\delta)^{\frac{1}{p}}}} d(t) = f(\floor{(d/\delta)^{\frac{1}{p}}}) \leq \sum_{t=1}^{\floor{(d/\delta)^{\frac{1}{p}}}} d \cdot t^{-p}$. If we add terms to this sum, then each term that would be added after $d\cdot (\floor{(d/\delta)^{\frac{1}{p}}})^{-p}$ will simply be $d(t) = \delta$, the minimum degree. We can upperbound the sum of any terms beyond $k > (d/\delta)^{\frac{1}{p}}$ by $\delta k$.
Thus we have the bound $f(k) \leq f(\floor{(d/\delta)^{\frac{1}{p}}}) + \delta k$.
In the case that $p=1$, the bound on the partial harmonic sum used in the proof of Theorem~\ref{thm:conv:gexpm} yields $f(\floor{(d/\delta)^{\frac{1}{p}}}) = f(\floor{d/\delta}) \leq d \cdot \sum_{t=1}^{\floor{d/\delta}} t^{-1} \leq  d( 1 + \log (d/\delta) ) \leq d(1+\log d)$, proving the $p=1$ case. If $p\neq 1$, we instead upperbound 
\[
\sum_{t=1}^{\floor{(d/\delta)^{\frac{1}{p}}}} t^{-p}
\leq d \left(1+\int_1^{(d/\delta)^{\frac{1}{p}}} x^{-p}\text{d$x$} \right)
= d\left( 1 + \tfrac{1}{1-p}\left( (d/\delta)^{\frac{1}{p}-1} - 1 \right) \right)
\leq \tfrac{1}{1-p}\left(d\cdot d^{\frac{1}{p}-1}\right)
\]
where the last inequality holds because $1 - \tfrac{1}{1-p} < 0$ for $p \in (0,1)$. Simplifying yields $\tfrac{1}{1-p}d^{\frac{1}{p}}$ as the final bound.
\end{proof}

We want to use this tighter bound on $f(k)$ to establish a tighter bound on $\| \vxj{r}{Nl}\|_1$. We can accomplish this using inequality \eqref{newresbd} if we first bound the sum $\sum_{k=b}^m 1/f(k)$ for constants $b,m$.

\begin{lemma}\label{lem:sum} In the notation of Lemma~\ref{lem:deg} we have
\begin{equation}\label{sum} \sum_{k=b}^m \tfrac{1}{f(k)} \geq \tfrac{1}{\delta}\log \left( \frac{\delta m + \delta + C_p}{\delta b+C_p}\right) \end{equation}
\end{lemma}
\begin{proof}
From Lemma~\ref{lem:deg} we can write $f(k) \leq C_p + \delta k$. Then $1/f(k) \geq 1/(C_p+\delta k)$, and so we have $\sum_{k=1}^m 1/f(k) \geq \sum_{k=1}^m 1/(C_p + \delta k)$.
Using a left-hand rule integral approximation, we get
\begin{equation}
\sum_{k=b}^m \tfrac{1}{C_p + \delta k}  \geq \int_b^{m+1} \tfrac{1}{C_p+\delta x}\text{d$x$} 
 = \tfrac{1}{\delta } \log \left( \frac{\delta m + \delta  + C_p }{\delta b + C_p} \right).
\end{equation}
\end{proof}
Plugging \eqref{sum} into \eqref{newresbd} yields, after some manipulation, our sublinearity result:
\begin{theorem}\label{thm:iterbnd}
In the notation of Lemma~\ref{lem:deg}, for a graph with power-law degree distribution with exponent $p \in (0,1]$, \gexpm attains $\|\vxj{r}{Nl}\|_1 < \eps$ in $Nl$ iterations if 
$
 l  > (3/\delta)(1/\eps)^{3\delta/2}C_p
$.
\end{theorem}
\begin{proof}
Before we can substitute \eqref{sum} into \eqref{newresbd}, we have to control the coefficients $s_i$. Note that the only entries in $\vr$ for which $s_k = (1-1/j)$ is equal to $1/2$ are the entries that correspond to the earliest time step, $j = 2$ (in the notation of Section~\ref{sec:linsys}). There are at most $d$ iterations that have a time step value of $j=2$, because only the neighbors of the starting node, node $c$, have non-zero entries in the $j = 2$ time step. Hence, every iteration other than those $d$ iterations must have $s_k \geq (1-1/j)$ with $j \geq 3$, which implies $s_k \geq \tfrac{2}{3}$. We cannot say \textit{which} $d$ iterations of the $Nl$ total iterations occur in time step $j=2$. However, the first $d$ values of $1/f(k)$ in $\sum_{k=1}^m s_k/f(k)$ are the largest in the sum, so by assuming those $d$ terms have the smaller coefficient (1/2 instead of 2/3), we can guarantee that
\begin{equation}\label{split}
- \sum_{k=1}^{l} \tfrac{s_{k}}{f(k)} < - \sum_{k=1}^d \tfrac{1/2}{f(k)} - \sum_{k=d+1}^{l} \tfrac{2/3}{f(k)} 
\end{equation}

To make the proof simpler, we omit the sum $\sum_{k=1}^d (1/2)/f(k)$ outright. From Corollary~\ref{cor:newresbd} we have
$\| \vxj{r}{Nl} \|_1  \leq \expof{ - \sum_{k=1}^{l} s_{k}/f(k) }$, which we can bound above with
$\expof{ - (2/3) \sum_{k=d+1}^{l} 1/f(k)}$, using inequality $\eqref{split}$. Lemma~\ref{lem:sum} allows us to upperbound the sum $-\sum_{k=d+1}^{l} 1/f(k)$, and simplifying yields
\begin{align*}
\| \vxj{r}{Nl} \|_1 \leq & \left( \frac{\delta l + \delta  + C_p}{\delta (d+1) + C_p} \right)^{-\frac{2}{3\delta }}. \label{resdelta}\\
\end{align*}
To guarantee $\| \vxj{r}{Nl} \|_1 < \eps$, then, it suffices to show that $\delta l + \delta  + C_p > (1/\eps)^{3\delta /2}(\delta d + \delta  +C_p)$. This inequality holds if $l$ is greater than $(1/\delta)(1/\eps)^{3\delta/2}(\delta d + \delta  +C_p)$.
Hence, it is enough for $l$ to satisfy
\[
 l \geq \tfrac{3}{\delta }(\tfrac{1}{\eps})^{\frac{3\delta }{2}}C_p.
\]
This last line requires the assumption $(\delta d + \delta  + C_p) < 3C_p$, which holds only if $\log d$ is larger than $\delta$ (in the case $p=1$), or if $d^{\frac{1}{p}-1}$ is larger than $\delta$ (in the case $p\neq 1$). Since we have been assuming that $d$ is a function of $n$ and $\delta$ is a constant independent of $n$, it is safe to assume this.
\end{proof}

With these technical lemmas in place, we are prepared to prove Theorem~\ref{thm:work} that gives the runtime bound for the \gexpm algorithm on graphs with a power-law degree distribution.

\paragraph{Proof of Theorem~\ref{thm:work}}
 Theorem \ref{thm:iterbnd} states that $ l \geq (3/\delta)(1/\eps)^{3\delta / 2}C_p$ will guarantee $\| \vxj{r}{Nl} \|_1 < \eps$. It remains to count the number of floating point operations performed in $Nl$ iterations.

Each iteration involves a vector add consisting of at most $d$ operations, and adding a column of $\mP$ to the residual, which consists of at most $d$ adds. Then, each entry that is added to the residual requires a heap update.
The heap updates at iteration $k$ involve at most $O(\log \nnz{k}  )$ work, since the residual heap contains at most $\nnz{k}$ non-zeros at that iteration. The heap is largest at the last iteration, so we can upperbound $\nnz{k} \leq \nnz{Nl}$ for all $k \leq Nl$. Thus, each iteration consists of no more than $d$ heap updates, and so $d\log ( \nnz{Nl} )$ total operations involved in updating the heap. Hence, after $Nl$ iterations, the total amount of work performed is upperbounded by $O(Nld \log \nnz{Nl})$.

After applying Lemmas \ref{lem:nnz} and \ref{lem:deg} we know the number of non-zeros in the residual (after $Nl$ iterations) will satisfy $ \nnz{Nl} \leq N f(l) < N \left( C_p + \delta l\right)$. Substituting in the expression for $l$ from \ref{thm:iterbnd} yields $\nnz{Nl} < N \left( C_p  + \delta  (3/\delta)(1/\eps)^{3\delta / 2} C_p \right )$. Upperbounding $C_p < (1/\eps)^{3\delta / 2} C_p$ allows us to write
\begin{equation}\nnz{Nl} < 4N(\tfrac{1}{\eps})^{\frac{3\delta }{2}} C_p. \label{worknnz}
\end{equation}

We can upperbound the work, work$(\eps)$, required to produce a solution with error $< \eps$, by using the inequalities in \eqref{worknnz} and in Theorem~\ref{thm:iterbnd}. We expand the bound
$ \text{work$(\eps)$} < Nld\log \nnz{Nl} $  to
\begin{align*}
&< Nd \left( (\tfrac{3}{\delta })(\tfrac{1}{\eps})^{\frac{3\delta}{2} }C_p \right) \cdot \log ( 4N(\tfrac{1}{\eps})^{\frac{3\delta }{2}} C_p ) \\
&< N \left( (\tfrac{3}{\delta })(\tfrac{1}{\eps})^{\frac{3\delta}{2} }d C_p \right)  \cdot \left( \log (4N) + \frac{3\delta }{2}\log(\tfrac{1}{\eps})+ \log(C_p)  \right),
\end{align*}
which we can upperbound with
$O\left( 3N(\tfrac{1}{\eps})^{\frac{3\delta }{2}} dC_p \cdot 4 \cdot\max\{ \log(d), \log(1/\eps) \} \right)$. This proves $\text{work$(\eps)$} =  O\left(N \left(1/\eps\right)^{3\delta /2} d C_p \cdot \max\{ \log(C_p), \log(1/\eps) \} \right)$. Replacing $N$ with the expression from Lemma \ref{lem:Nbound} yields the bound on total work given in Theorem~\ref{thm:work}.

%% file: sec-experiments.tex
\section{Experimental Results}\label{sec:exp}
Here we evaluate our algorithms' accuracy and speed for large real-world and synthetic networks. 

\paragraph{Overview}
To evaluate accuracy, we examine how well the \gexpmq function identifies the largest entries of the true solution vector. This is designed to study how well our approximation would work in applications that use large-magnitude entries to find important nodes (Section~\ref{sec:exp:acc}). We find a tolerance of $10^{-4}$ is sufficient to accurately find the largest entries at a variety of scales. We also provide more insight into the convergence properties of \expmimv by measuring the accuracy of the algorithm as the size $z$ of its heap varies (Section~\ref{sec:exp:expmimv}). Based on these experiments, we recommend setting the subset size for that algorithm to be near $(\text{nnz}(\mP) / n)$ times the number of large entries desired.

We then study how the algorithms scale with graph size. We first compare their runtimes on real-world graphs with varying sizes (Section~\ref{sec:exp:run}). The edge density and maximum degree of the graph will play an important role in the runtime. This study illustrates a few interesting properties of the runtime that we examine further in an experiment with synthetic forest-fire graphs of up to a billion edges. Here, we find that the runtime scaling grows roughly as $d^2$, as predicted by our theoretical results.

\paragraph{Real-world networks}
The datasets used are summarized in Table \ref{tab:datasets}. They include a version of the flickr graph from~\cite{Bonchi-2012-fast-katz} containing just the largest strongly-connected component of the original graph; dblp-2010 from~\cite{boldi-2011-layered}, itdk0304 in~\cite{Caida-2005-network}, ljournal-2008 from~\cite{boldi-2011-layered, Chierichetti:2009:CSN:1557019.1557049}, twitter-2010~\cite{Kwak2010-Twitter} 
webbase-2001 from~\cite{hirai2000-webbase, boldi2005-codes},
 and the friendster graph in~\cite{Yang-2012-ground-truth}.

\begin{table}
\centering
\caption{The real-world datasets we use in our experiments span three orders of magnitude in size.} \label{tab:datasets}
\begin{tabularx}{\linewidth}{Xrrrrr} \toprule
Graph& $|V|$ &$\nnz{\mP}$ &  $\nnz{\mP}/|V|$ & $d$ & $\sqrt{|V|}$\\ \midrule 
\texttt{itdk0304}&       190,914   &  1,215,220 & 6.37 & 1,071 & 437 \\ %
\texttt{dblp-2010}&        226,413   &  1,432,920 & 6.33 & 238 & 476 \\ %
\texttt{flickr-scc}&        527,476   &  9,357,071 & 17.74 & 9,967 & 727 \\  %
\texttt{ljournal-2008} &  5,363,260 &  77,991,514 & 14.54 & 2,469 & 2,316 \\ %
\texttt{webbase-2001} & 118,142,155   & 1,019,903,190 & 8.63 & 3,841 & 10,870 \\%
\texttt{twitter-2010} & 33,479,734 & 1,394,440,635 & 41.65 & 768,552 & 5,786 \\%
\texttt{friendster} &  65,608,366   & 3,612,134,270 & 55.06 & 5,214 & 8,100 \\%
 \bottomrule
\end{tabularx}
\end{table}

\paragraph{Implementation details}
All experiments were performed on either a dual processor Xeon e5-2670 system with 16 cores (total) and 256GB of RAM or a single processor Intel i7-990X, 3.47 GHz CPU and 24 GB of RAM. Our algorithms were implemented in C++ using the Matlab MEX interface. All data structures used are memory-efficient: the solution and residual are stored as hash tables using Google's \texttt{sparsehash} package. The precise code for the algorithms and the experiments below are available via \texttt{https://www.cs.purdue.edu/homes/dgleich/codes/nexpokit/}.

\paragraph{Comparison}
We compare our implementation with a state-of-the-art Matlab function for computing the exponential of a matrix times a vector, \expmv, which uses a Taylor polynomial approach~\cite{Al-Mohy-2011-exponential}. We customized this method with the knowledge that $\normof[1]{\mP} = 1$. This single change results in a great improvement to the runtime of their code. In each experiment, we use as the ``true solution'' the result of a call to \expmv using the `single' option, which guarantees a relative backward error bounded by $2^{-24}$, or, for smaller problems, we use a Taylor approximation with the number of terms predicted by Lemma~\ref{lem:deg}.

\subsection{Accuracy on Large Entries}\label{sec:exp:acc}
When both $\gexpm$ and $\gexpmq$ terminate, they satisfy a 1-norm error of $\eps$. Many applications do not require precise solution \emph{values} but instead would like the correct \emph{set} of large-magnitude entries.
To measure the accuracy of our algorithms in identifying these large-magnitude entries, we examine the set precision of the approximations. Recall that the precision of a set $T$ that approximates a desired set $S$ is the size of their intersection divided by the total size: $|S \cap T|/|S|$. Precision values near 1 indicate accurate sets and values near 0 indicate inaccurate sets.  We show the precision as we vary the solution tolerance $\eps$ for the $\gexpmq$ method in Figure~\ref{fig:accuracy_vs_error}. The experiment we conduct is to take a graph, estimate the matrix exponential for 100 vertices (trials) for our method with various tolerances $\eps$, and compare the sets of the top $100$ vertices \emph{that are not neighbors of the seed node} between the true solution and the solution from our algorithm. We remove the starting node and its neighbors because these entries are \emph{always} large, so accurately identifying them is a near-guarantee. The results show that, in median performance, we get the top 100 set completely correct with $\eps = 10^{-4}$ for the small graphs. 

\begin{figure}[tp]
\includegraphics[width=0.5\linewidth]{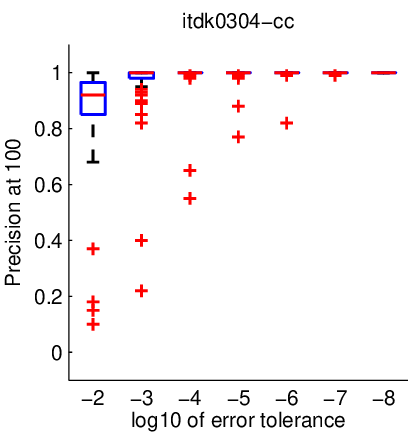}%
\includegraphics[width=0.5\linewidth]{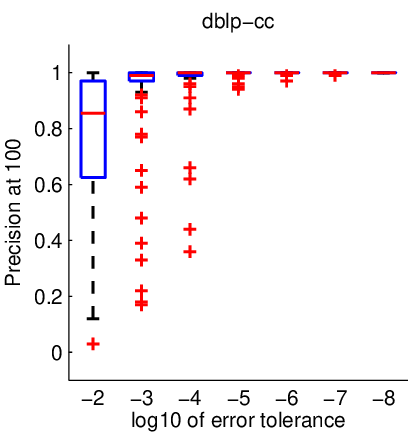}
\includegraphics[width=0.5\linewidth]{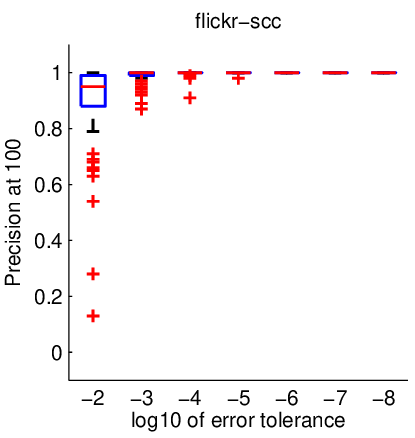}%
\includegraphics[width=0.5\linewidth]{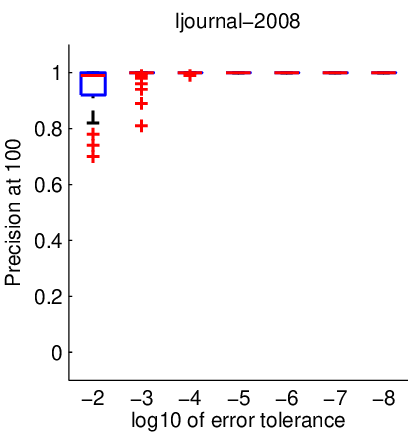}
\caption{We ran our method $\gexpmq$ for 100 different seed nodes as we varied the tolerance $\eps$. These four figures show box-plots over the 100 trials of the set precision scores of the top-100 results compared to the true top-100 set. These illustrate that a tolerance of $10^{-4}$ is sufficient for high accuracy.}
\label{fig:accuracy_vs_error}
\end{figure}

Next, we study how the work performed by the algorithm $\gexpmq$ scales with the accuracy. For this study, we pick a vertex at random and vary the maximum number of iterations performed by the \gexpmq algorithm. Then, we look at the set precision for the top-$k$ sets. The horizontal axis in Figure~\ref{fig:steps_vs_accuracy} measures the number of effective matrix-vector products based on the number of edges explored divided by the total number of non-zeros of the matrix. Thus, one matrix-vector product of work corresponds with looking at each non-zero in the matrix once.

The results show that we get good accuracy for the top-$k$ sets up to $k=1000$ with a tolerance of $10^{-4}$, \emph{and} converge in less than one matrix-vector product, with the sole exception of the flickr network. This network has been problematic for previous studies as well~\cite{Bonchi-2012-fast-katz}. Here, we note that we get good results in less than one matrix-vector product, but we do not detect convergence until after a few matrix-vector products worth of work. 

\begin{figure}[tp]
\centering
\includegraphics[width=0.5\linewidth]{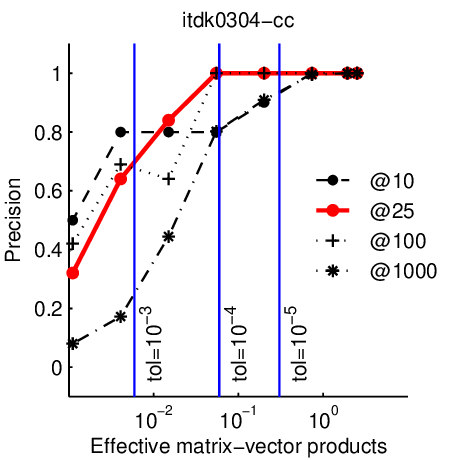}%
\includegraphics[width=0.5\linewidth]{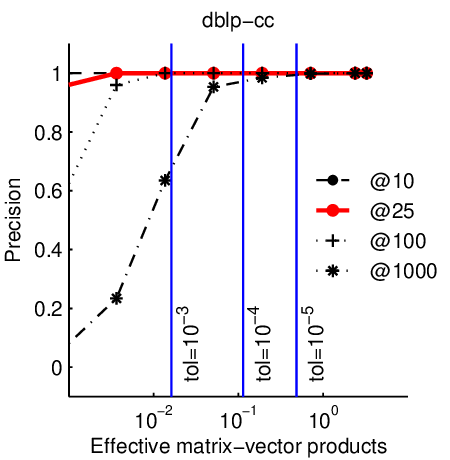}
\includegraphics[width=0.5\linewidth]{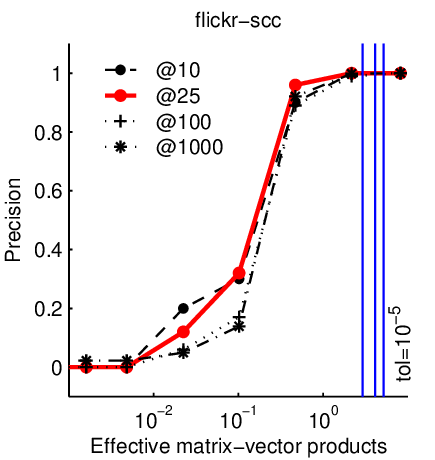}%
\includegraphics[width=0.5\linewidth]{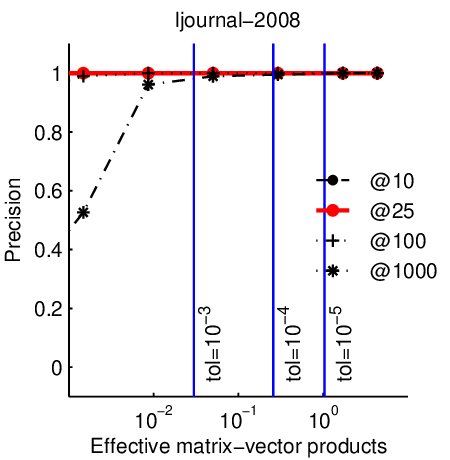}
\caption{For a single vertex seed, we plot how the precision in the top-$k$ entries varies with the amount of work for the $\gexpmq$ algorithm. The four pictures show the results of the four small graphs. We see that for every graph except flickr, the results are good with a tolerance of $10^{-4}$, which requires less than one mat-vec worth of work. }
\label{fig:steps_vs_accuracy}
\end{figure}

\subsubsection{Accuracy \textit{\&} non-zeros with incomplete matrix-vector products}\label{sec:exp:expmimv}
The previous studies explored the accuracy of the $\gexpmq$ method. Our cursory experiments showed that $\gexpm$ behaves similarly because it also achieves an $\eps$ error in the 1-norm. In contrast, the \expmimv method is rather different in its accuracy because it prescribes only a total size of intermediate heap; we are interested in accuracy as we let the heap size increase.

The precise experiment is as follows. For each graph, repeat the following: first, compute 50 node indices uniformly at random. For each node index, use \expmimv to compute $\epec$ using different values for the heap size parameter: $z = 100$, 200, 500, 1000, 2000, 5000, 10000. Figure~\ref{fig:expmimv:acc} displays the median of these 50 trials for each parameter setting for both the 1-norm error and the top-1000 set precision of the \expmimv approximations. (The results for top-100 precision, as in the previous study, were effectively the same.) 
The plot of the 1-norm error in Figure~\ref{fig:expmimv:acc} (left) displays clear differences for the various graphs, yet there are pairs of graphs that have nearby errors (such as itdk0304 and dblp).  The common characteristic for each pair appears to be the edge density of the graph. We see this effect more strongly in the right plot where we look at the precision in the top-1000 set.

Again, set precision improves for all datasets as more non-zeros are used. If we normalize by edge density (by dividing the number of non-zeros used by the edge density of each graph) then the curves cluster. Once the ratio (non-zeros used / edge density) reaches 100, \expmimv attains a set precision over 0.95 for \emph{all} datasets on the 1,000 largest-magnitude nodes, regardless of the graph size. We view this as strong evidence that this method should be useful in many applications where precise numeric values are not required.

\begin{figure}[t]
\includegraphics[width=0.5\linewidth]{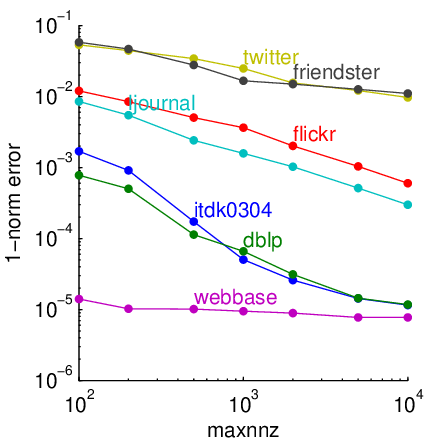}%
\includegraphics[width=0.5\linewidth]{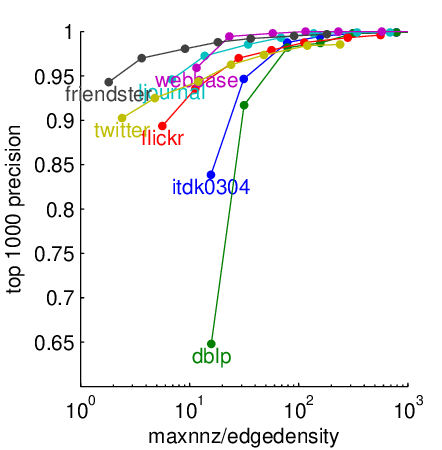}

\caption{Here we display the performance of the $\expmimv$ method. The left figure shows the 1-norm error compared with the number of non-zeros retained in the matrix-vector products on our set of graphs. The groups of curves with similar convergence have comparable edge densities. The right figure shows how the set precision converges as we increase the number of non-zeros, relative to the edge density of the graph ($\nnz{\mP}/n$). 
}
\label{fig:expmimv:acc}
\end{figure}

\subsection{Runtime \textit{\&} Input-size}\label{sec:exp:run}
Because the algorithms presented here are intended to be fast on large, sparse networks, we continue our study by investigating how their speed scales with data size. Figure~\ref{fig:runtimes} displays the median runtime for each graph, where the median is taken over 100 trials for the smaller graphs, and 50 trials for the twitter and friendster datasets. Each trial consists of computing a column $\expof{\mP}\ve_c$ for randomly chosen $c$. All algorithms use a fixed 1-norm error tolerance of $10^{-4}$; the \expmimv method uses 10,000 non-zeros, which may not achieve our desired tolerance, but identifies the right set with high probability, as evidenced in the experiment of Section~\ref{sec:exp:expmimv}. 

\begin{figure}
\centering
\includegraphics[width=0.8\linewidth]{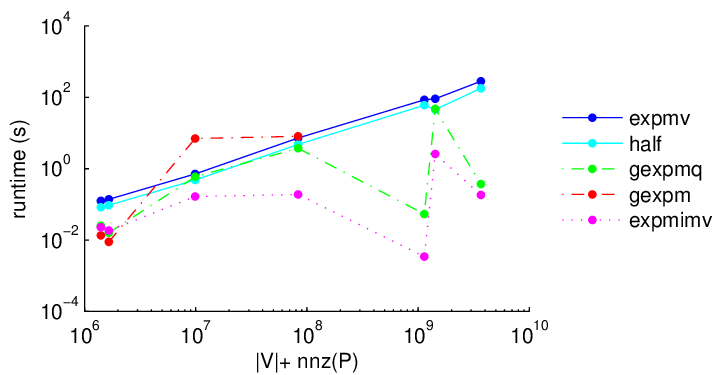}%
\caption{The median runtime of our methods for the seven graphs over 100 trials (only 50 trials for the largest two datasets), compared with the method \texttt{expmv} of~\citet{Al-Mohy-2011-exponential} using the \texttt{single} and \texttt{half} accuracy settings (which we label as \texttt{expmv} and \texttt{half}, respectively). The coordinate relaxation methods have highly variable runtimes, but can be very fast on graphs such as webbase (the point nearest $10^{9}$ on the $x$-axis). We did not run the \gexpm function for matrices larger than the livejournal graph.}
\label{fig:runtimes}
\end{figure}

\begin{figure}
\includegraphics[width=0.33\linewidth]{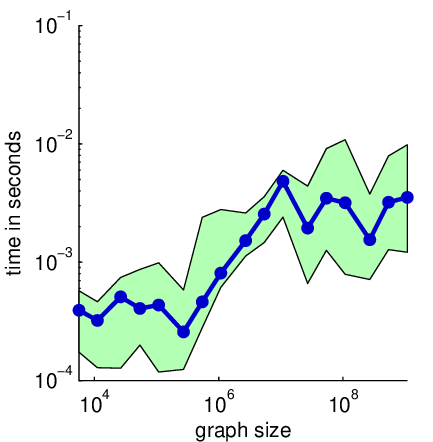}%
\includegraphics[width=0.33\linewidth]{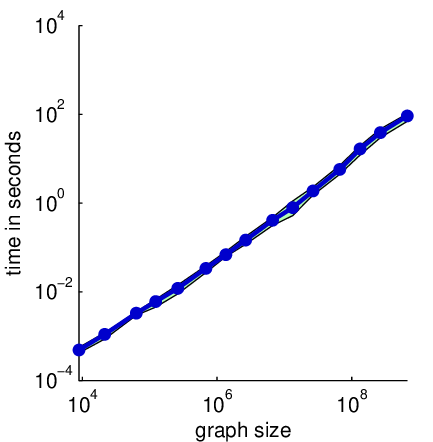}%
\includegraphics[width=0.33\linewidth]{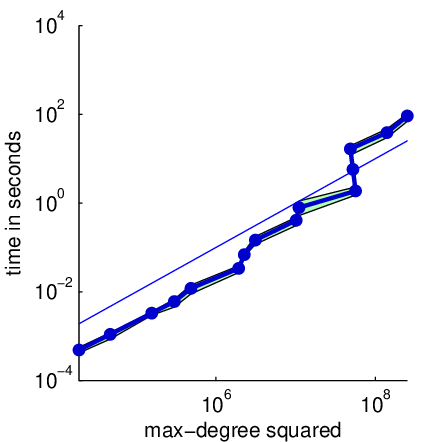}
\caption{The distribution of runtimes for the \gexpm method on two forest-fire graphs (left: $p_f=0.4$, middle, $p_f=0.48$) of various graph sizes, where graph size is computed as the sum of the number of vertices and the number of non-zeros in the adjacency matrix. The thick line is the median runtime over 50 trials, and the shaded region shows the 25\% to 75\% quartiles (the shaded region is very tight for the second two figures). The final plot (right) shows the relationship between the max-degree squared and the runtime in seconds. This figure shows that the runtime scales in a nearly linear relationship with the max-degree squared, as predicted by our theory. The large deviations from the line of best fit might be explained by the fact that only a single forest-fire graph was generated for each graph size.
}
\label{fig:ff-scaling}
\end{figure}

\subsection{Runtime scaling}
The final experimental study we conduct attempts to better understand the runtime scaling of the \gexpmq method. This method yields a prescribed accuracy $\eps$ more rapidly than \gexpm, but the study with real-world networks did not show a clear relationship between error and runtime. We conjecture that this is because the edge density varies too much between our graphs. Consequently, we study the runtime scaling on forest-fire synthetic graphs~\cite{Leskovec-2007-densification}. We use a symmetric variation on the forest-fire model with a single ``burning'' probability. We vary the number of vertices generated by the model to get graphs ranging from around 10,000 vertices to around 100,000,000 vertices.

The runtime distributions for burning probabilities, $p_f$, of $0.4$ and $0.48$ are shown in the left two plots of Figure~\ref{fig:ff-scaling}. With $p_f=0.48$, the graph is fairly dense -- more like the friendster network -- whereas the graph with $p_f=0.4$ is highly sparse and is a good approximation for the webbase graph. Even with billions of edges, it takes less than $0.01$ seconds for $\gexpmq$ to produce a solution with 1-norm error $\eps = 10^{-4}$ on this sparse graph. For $p_f=0.48$ the runtime grows with the graph size.

We find that the scaling of $d^2$ seems to match the empirical scaling of the runtime (right plot), which is a plausible prediction based on Theorem~\ref{thm:work}. (Recall that one of the log factors in the bound of Theorem~\ref{thm:work} arose from the heap updates in the \gexpm method.) These results show that our method is extremely fast when the graph is sufficiently sparse, but does slow down when running on networks with higher edge density.

%% file: sec-conclusions.tex
\section{Conclusions \textit{\&} Future Work}\label{sec:con}
The algorithms presented in this paper compute a column of the matrix exponential of sparse matrices $\mP$ satisfying $\|\mP\|_1 \leq 1$. They range from \gexpm, with its strong theoretical guarantees on accuracy and runtime, to \gexpmq, which drops the theoretical runtime bound but is empirically faster and provably accurate, to \expmimv, where there is no guarantee about accuracy but there is an even smaller runtime bound. We also showed that they outperform a state of the art Taylor method, \expmv, to compute a column of the matrix exponential experimentally. This suggests that these methods have the potential to become methods of choice for computing columns of the matrix exponential on social networks.

We anticipate that our method will be useful in scenarios where the goal is to compare the matrix exponential with other network measures -- such as the personalized PageRank vector. We have used a variant of the ideas presented here, along with a new runtime bound for a degree-weighted error, to perform such a comparison for the task of community detection in networks~\cite{Kloster-preprint-hkrelax}.

\paragraph{Functions beyond the exponential.}
We believe we can generalize the results here to apply to a larger class of inputs: namely, sparse matrices $\mA$ satisfying $\| \mA \|_1 \leq c$ for small $c$. Furthermore, all three algorithms described in this paper can be generalized to work for functions other than $e^x$. In our future work, we also plan to explore better polynomial approximations of the exponential~\cite{Orecchia-2012-exponential}.

\paragraph{Improved implementations.} Because of the slowdown due to the heap updates in \gexpm, we hope to improve on our current heap-based algorithm by implementing a new data structure that can provide fast access to large entries as well as fast update to and deletion of entries. One of the possibilities we wish to explore is a Fibonacci heap. We also plan to explore parallelizing the algorithms using asynchronous methods. Recent analysis suggests that strong, rigorous runtime guarantees are possible~\cite{Avron-2014-asynchronous}.

\paragraph{New analysis.} Finally, we hope to improve on the analysis for \expmimv. Namely, we believe there is a rigorous relationship between the input graph size, the non-zeros retained, the Taylor degree selected, and the error of the solution vector produced by \expmimv. Recently, one such method was rigorously analyzed~\cite{Deshpande-2013-cliques}, which helped to establish new bounds on the planted clique problem. 

%% file: sec-appendix.tex
\section{Appendix -- Proofs}\label{sec:app}

\def\thetheorem{\ref{lem:Nbound}}
\begin{lemma*}
Let $\mP$ and $\vb$ satisfy $\|\mP \|_1, \|\vb\|_1 \leq 1$. Then choosing the degree, $N$, of the Taylor approximation, $T_N(\mP)$, such that $N \geq 2\log(1/\eps)$ and $N \geq 3$ will guarantee
\[
\|\expof{\mP} \vb - T_N(\mP)\vb \|_1 \leq \eps
\]
\end{lemma*}

\begin{proof}
We first show that the degree $N$ Taylor approximation satisfies
\begin{equation}
\|\expof{\mP} \vb - T_N(\mP)\vb \|_1 \leq \tfrac{1}{N! N}. \label{taydeg}
\end{equation}
To prove this, observe that our approximation's remainder, $\|\expof{\mP} \vb - T_N(\mP)\vb \|_1$, equals
$\normof*[1]{\sum_{k=N+1}^{\infty} \mP^k \ve_c/k! }$. Using the triangle inequality we can upperbound this by $\sum_{k=N+1}^{\infty} \| \mP^k \|_1 \|\ve_c \|_1 / k! $.We then have
\[
\|\expof{\mP} \vb - T_N(\mP)\vb \|_1 \leq \sum_{k=N+1}^{\infty} \tfrac{1}{k!}
\]
because $\|\mP\|_1 \leq 1$ and $\|\ve_c\|_1 = 1$. By factoring out $1/(N+1)!$ and majorizing $(N+1)!/(N+1+k)! \leq 1/(N+1)^k$ for $k \geq 0 $, we finish:
\begin{equation}
 \|\expof{\mP} \vb - T_N(\mP)\vb \|_1 \leq \lrp{\tfrac{1}{(N+1)!}}\sum_{k=0}^{\infty} \lrp{\tfrac{1}{N+1}}^k = \tfrac{1}{(N+1)!} \tfrac{N+1}{N}
\end{equation}
where the last step substitutes the limit for the convergent geometric series.

Next, we prove the lemma. We will show that $2\log (N!) > N \log N$, then use this to relate $\log(N!N)$ to $\log(\eps)$. First we write $2\log(N!) = 2\cdot\sum_{k=0}^{N-1} \log(1+k) = \sum_{k=0}^{N-1} \log(1+k) + \sum_{k=0}^{N-1} \log(1+k)$. By noting that $\sum_{k=0}^{N-1} \log(1+k) = \sum_{k=0}^{N-1} \log(N-k)$, we can express $2 \log(N!) = \sum_{k=0}^{N-1} \log(k+1) + \sum_{k=0}^{N-1} \log(N-k)$, which is equal to $\sum_{k=0}^{N-1} \log\left((k+1)(N-k)\right)$. Finally, $(k+1)(N-k) = N + Nk -k^2 - k = N + k(N-k-1)\geq N$ because $N \geq k+1$, and so 
\begin{equation}\label{logfact}
2\log(N!) \geq \sum_{k=0}^{N-1} \log(N) = N\log(N).
\end{equation}
By the first claim we know that $1/N!N < \eps$ guarantees the error we want, but for this inequality to hold it is sufficient to have $\log(N!N) > \log(1/\eps)$. Certainly if $\log(N!) > \log(1/\eps)$ then $\log(N!N) > \log(1/\eps)$ holds, so by \eqref{logfact} it suffices to choose $N$ satisfying $N\log(N) > 2\log(1/\eps)$. Finally, for $N\geq 3$ we have $\log(N) > 1$, and so Lemma~\ref{lem:Nbound} holds for $N\geq 3$.
\end{proof}

\def\thetheorem{\ref{lem:Minverse}}
\begin{lemma*}
Let $\mM = (\mI_{N+1} \kron \mI_n - \mS \kron \mA )$, where $\mS$ denotes the $(N+1) \times (N+1)$ matrix with first sub-diagonal equal to $[1/1, 1/2, ... , 1/N]$, and $\mI_k$ denotes the $k\times k $ identity matrix. Then
$\mM\inv = \sum_{k=0}^{N} \mS^k \kron \mA^k.$
\end{lemma*}
\begin{proof}
Because $\mS$ is a subdiagonal matrix, it is nilpotent, with $\mS^{N+1} = 0$. This implies that $\mS\kron \mA$ is also nilpotent, since $(\mS\kron \mA)^{N+1} = \mS^{N+1}\kron\mA^{N+1} = 0\kron\mA = 0$. Thus, we have
\begin{align*}
\mM\left(\sum_{k=0}^{N} \mS^k \kron \mA^k\right) &= (\mI - \mS\kron\mA)\left( \sum_{k=0}^{N} \mS^k \kron \mA^k \right)\\
&= \mI - (\mS\kron\mA)^{N+1} & \text{the sum telescopes}
\end{align*}
which is $\mI$. This proves $(\sum_{k=0}^{N} \mS^k \kron \mA^k)$ is the inverse of $\mM$.
\end{proof}

\def\thetheorem{\ref{lem:linsyserror}}
\begin{lemma*}
Consider an approximate solution $\hat{\vv} = [ \hvv_0; \hvv_1; \cdots; \hvv_N] $ to the linear system 
\[(\mI_{N+1} \kron \mI_n - \mS \kron \mA )[ \vv_0; \vv_1; \cdots; \vv_N] = \ve_1 \kron \ve_c.
\] 
Let $\vx = \sum_{j=0}^N \hat{\vv}_j$, let $T_N(x)$ be the degree $N$ Taylor polynomial for $e^x$, and define $\psi_j(x) = \sum_{m=0}^{N-j} \tfrac{j!}{(j+m)!} x^m$. Define the residual vector $\vr = [\vr_0 ; \vr_1 ; \ldots ; \vr_N]$ by $\vr: = \ve_1 \kron \ve_c - (\mI_{N+1} \kron \mI_n - \mS \kron \mA ) \hat{\vv}$. Then the error vector $T_N(\mA)\ve_c - \vx$ can be expressed
\[
T_N(\mA)\ve_c - \vx = \sum_{j=0}^N \psi_j(\mA)\vr_j.
\]
\end{lemma*}
\begin{proof}
Recall that $\vv = [ \vv_0; \vv_1; \cdots; \vv_N]$ is the solution to equation \eqref{linsys}, and our approximation is $\hvv = [ \hvv_0; \hvv_1; \cdots; \hvv_N]$. We showed in Section \ref{sec:linsys} that the error $T_N(\mA)\ve_c - \vx$ is in fact the sum of the error blocks $\vv_j - \hvv_j$. Now we will express the error blocks $\vv_j - \hvv_j$ in terms of the residual blocks of the system \eqref{linsys}, i.e. $\vr_j$.

The following relationship between the residual vector and solution vector always holds:
$\vr = \eoe{1}{c} - \mM \hvv$,
so pre-multiplying by $\mM\inv$ yields
$ \mM\inv \vr = \vv - \hvv,$
because $\vv = \mM\inv \eoe{1}{c}$ exactly, by definition of $\vv$. Note that $ \mM\inv \vr = \vv - \hvv$
is the error vector for the linear system \eqref{linsys}. Substituting the expression for $\mM\inv$ in Lemma~\ref{lem:Minverse} yields
\begin{equation}\label{errorblocks}
\bmat{ \vv_0 -\hvv_0 \\ \vv_1 -\hvv_1 \\ \vdots \\ \vv_N -\hvv_N } =  \lp \sum_{k=0}^N \mS^k \kron \mA^k \rp \bmat{ \vr_0 \\ \vr_1 \\ \vdots \\ \vr_N }.
\end{equation}

Let $\ve$ be the vector of all 1s of appropriate dimension. Then observe that pre-multiplying equation \eqref{errorblocks} by $\lrp{\ve^T\kron \mI}$ yields, on the left-hand side, $\sum_{j=0}^N (\vv_j - \hvv_j)$. Now we can accomplish our goal of expressing $\sum_{j=0}^N (\vv_j - \hvv_j)$ in terms of the residual blocks $\vr_j$ by expressing the right-hand side $\lrp{\ve^T\kron \mI}\lp \sum_{k=0}^N \mS^k \kron \mA^k \rp \vr$ in terms of the blocks $\vr_j$. So next we consider the product of a fixed block $\vr_{j-1}$ with a particular term $(\mS^k\kron \mA^k)$. Note that, because $\vr_{j-1}$ is in block-row $j$ of $\vr$, it multiplies with only the block-column $j$ of $(\mS^k\kron \mA^k)$, so we now examine the blocks in block-column $j$ of $(\mS^k\kron \mA^k)$.

Because $\mS$ is a subdiagonal matrix, there is only one non-zero in each column of $\mS^k$, for each $k = 0, ... , N$. As mentioned in Section~\ref{sec:gexpm}, $\mS\ve_j = \ve_{j+1}/j$ when $j<N+1$, and 0 otherwise. This implies that
\[
\mS^k\ve_j =
 \begin{cases} 
      \hfill \frac{(j-1)!}{(j-1+k)!}\ve_{j+k} ,  \hfill & \text{ if } 0\leq k \leq N+1-j \\
      \hfill 0 , \hfill & \text{otherwise.} \\
  \end{cases}
\] 
Thus, block-column $j$ of $(\mS^k \kron \mA^k)$ contains only a single non-zero block, $(j-1)!\mA^k/(j-1+k)!$, for each $k = 0, ... , N+1-j$. Hence, summing the $n \times n$ blocks in block-column $j$ of all powers $(\mS^k \kron \mA^k)$ for $k=0, ... , N$ yields
\begin{equation}\label{rescoeff}
\sum_{k=0}^{N+1-j} \tfrac{(j-1)! }{(j-1+k)!} \mA^k
\end{equation}
as the matrix coefficient of the term $\vr_{j-1}$ in the expression
$(\ve^T \kron \mI)(\sum_{k=0}^N \mS^k \kron \mA^k) \vr$.
Thus, we have
\[
(\ve^T \kron \mI)\lrp{\sum_{k=0}^N \mS^k \kron \mA^k} \vr = \sum_{j=1}^{N+1} \lrp{ \sum_{k=0}^{N+1-j}\tfrac{(j-1)!}{(j-1+k)!)} \mA^k }\vr_{j-1}
\]
Finally, reindexing so that the outer summation on the right-hand side goes from $j=0$ to $N$, then substituting our definition for $\psi_j(\mA) = \sum_{k=0}^{N-m} \frac{j!}{(j+k)!} \mA^k$, we have that 
$\sum_{j=0}^N (\vv_j - \hvv_j) = \sum_{j=0}^N \psi_j(\mA) \vr_j,$ as desired.
\end{proof}

\begin{lemma*}[From~\cite{Kloster-2013-nexpokit}]
Let $\psi_j(x) = \sum_{m=0}^{N-j} \frac{j!}{(j+m)!} x^m$. Then $\psi_j(1) \leq \psi_0(1) \leq \exp(1)$.
\end{lemma*}
\begin{proof}
By definition, $\psi_j(1) = \sum_{m=0}^{N-j} \frac{j!}{(j+m)!}$ and $\psi_{j+1}(1) = \sum_{m=0}^{N-j-1} \frac{(j+1)!}{(j+1+m)!}$. Note that the sum for $\psi_j(1)$ has more terms, and the general terms of the two summations satisfy $\frac{j!}{(j+m)!} \geq \frac{(j+1)!}{(j+1+m)!}$ because multiplying both sides by $\frac{(j+m)!}{j!}$ yields $1 \geq \frac{j+1}{j+1+m}$. Hence $\psi_j(1) \geq \psi_{j+1}(1)$ for $j=0, ..., N-1$, and so the statement follows.

To see that $\psi_0(1) \leq \exp(1)$, note that $\psi_0(1)$ is the degree $N$ Taylor polynomial expression for $\exp(1)$, which is a finite approximation of the Taylor series, an infinite sum of positive terms; hence, $\psi_0(1) \leq \exp(1)$.
\end{proof}